\documentclass[reqno]{amsart}
\usepackage{amsmath,amsfonts,amsthm,color,amssymb,bm}
\usepackage{color}
\usepackage{graphicx}
\usepackage{hyperref}
\usepackage[margin=3cm]{geometry}
\usepackage{stmaryrd}
\usepackage{enumerate}
\usepackage{float}

\usepackage[textsize=tiny]{todonotes}

\newtheorem{theorem}{Theorem}[section]
\newtheorem{corollary}[theorem]{Corollary}
\newtheorem{lemma}[theorem]{Lemma}
\newtheorem{proposition}[theorem]{Proposition}

\theoremstyle{definition}
\newtheorem{definition}[theorem]{Definition}
\newtheorem{remark}[theorem]{Remark}

\newcommand{\D}{\mathrm{d}}
\newcommand{\E}{\mathrm{e}}
\newcommand{\ds}{\mathrm{d}s}
\newcommand{\dt}{\mathrm{d}t}

\newcommand{\dr}{\mathrm{d}r}
\newcommand{\dx}{\mathrm{d}x}
\newcommand{\dy}{\mathrm{d}y}

\newcommand{\half}{\frac{1}{2}}

\newcommand{\norm}[1]{\left\|#1\right\|}
\newcommand{\notthis}[1]{}

\newcommand{\DD}{\mathbb{D}}
\newcommand{\EE}{\mathbb{E}}
\newcommand{\NN}{\mathbb{N}}
\newcommand{\PP}{\mathbb{P}}
\newcommand{\RR}{\mathbb{R}}
\newcommand{\TT}{\mathbb{T}}
\newcommand{\mf}{\mathfrak{m}}
\newcommand{\Mf}{\mathfrak{M}}
\newcommand{\ut}{\mathfrak{u}}

\newcommand{\fp}{\mathfrak{p}}
\newcommand{\fq}{\mathfrak{q}}
\newcommand{\Kf}{\mathfrak{K}}
\newcommand{\Gf}{\mathfrak{G}}
\newcommand{\Hf}{\mathfrak{H}}

\newcommand{\Cc}{\mathcal{C}}
\newcommand{\Ee}{\mathcal{E}}

\newcommand{\Ll}{\mathcal{L}}
\newcommand{\Nn}{\mathcal{N}}

\newcommand{\Ff}{\mathcal{F}}

\newcommand{\Ii}{\mathcal{I}}
\newcommand{\Ss}{\mathcal{S}}

\newcommand{\Df}{\mathrm{D}}
\newcommand{\xx}{\mathbf{x}}

\newcommand{\mb}{\mathbf{m}}
\newcommand{\BS}{\mathrm{BS}}
\newcommand{\BSB}{\overleftarrow{\mathrm{BS}}}
\newcommand{\BST}{\widetilde{\mathrm{BS}}}
\newcommand{\VIX}{\mathrm{VIX}}
\newcommand{\phib}{\boldsymbol\phi}
\newcommand{\Jb}{\boldsymbol J}
\newcommand{\rhob}{\boldsymbol\rho}

\newcommand{\Wb}{\mathbf{W}}

\newcommand{\Thb}{\boldsymbol\Theta}

\newcommand{\BSh}{\widehat{\BS}}
\newcommand{\Ph}{\widehat{\mathrm{P}}}
\newcommand{\lef}{\leftarrow}
\newcommand{\at}{\widetilde\alpha}
\newcommand{\bt}{\widetilde\beta}
\newcommand{\Hm}{H_-}
\newcommand{\Hp}{H_+}
\newcommand{\wf}{\mathfrak{w}}
\newcommand{\chib}{\overline{\chi}}
\newcommand{\rrho}{\overline{\rho}}
\newcommand{\ep}{\varepsilon}
\newcommand{\dto}{\downarrow}
\newcommand{\epst}{\ep_{1}}
\newcommand{\epsz}{\ep_{0}}

\newcommand{\ind}{1\hspace{-2.1mm}{1}}

\newcommand{\pd}{\partial}
\newcommand{\pdx}{\pd_{x}}
\newcommand{\pdxx}{\pd_{xx}}
\newcommand{\pds}{\pd_{s}}
\newcommand{\pdt}{\pd_{t}}
\newcommand{\pdu}{\pd_{u}}
\newcommand{\pdsi}{\pd_{\sigma}}
\newcommand{\pdy}{\pd_{y}}

\newcommand{\pdxy}{\pd_{xy}}
\newcommand{\pdk}{\pd_{k}}

\newcommand{\pdxk}{\pd_{xk}}

\newcommand{\dxf}{\pdx f}

\newcommand{\dxxf}{\pdxx f}

\newcommand{\HBarlg}{$\boldsymbol{(\overline{\mathrm{H}}}^{\lambda\gamma}\boldsymbol{)}$}
\newcommand{\HBarl}{$\boldsymbol{(\overline{\mathrm{H}}}^{\lambda}\boldsymbol{)}$}
\newcommand{\HAll}{$\boldsymbol{(\mathrm{H}_{12345})}$}
\newcommand{\HOT}{$\boldsymbol{(\mathrm{H}_{123})}$}
\newcommand{\Hl}{$\boldsymbol{(\mathrm{H}}_{\boldsymbol{67}}^{\lambda}\boldsymbol{)}$}
\newcommand{\Hg}{$\boldsymbol{(\mathrm{H}}_{\boldsymbol{89}}^{\gamma}\boldsymbol{)}$}

\newcommand{\HOne}{$\boldsymbol{(\mathrm{H}_{1})}$}
\newcommand{\HTwo}{$\boldsymbol{(\mathrm{H}_{2})}$}
\newcommand{\HThree}{$\boldsymbol{(\mathrm{H}_{3})}$}
\newcommand{\HFour}{$\boldsymbol{(\mathrm{H}_{4})}$}
\newcommand{\HFive}{$\boldsymbol{(\mathrm{H}}_{\boldsymbol{5}}\boldsymbol{)}$}
\newcommand{\HSix}{$\boldsymbol{(\mathrm{H}}_{\boldsymbol{6}}^{\lambda}\boldsymbol{)}$}
\newcommand{\HSeven}{$\boldsymbol{(\mathrm{H}}_{\boldsymbol{7}}^{\lambda}\boldsymbol{)}$}
\newcommand{\HEight}{$\boldsymbol{(\mathrm{H}}_{\boldsymbol{8}}^{\gamma}\boldsymbol{)}$}
\newcommand{\HNine}{$\boldsymbol{(\mathrm{H}}_{\boldsymbol{9}}^{\gamma}\boldsymbol{)}$}

\newcommand{\COne}{$\boldsymbol{(\mathrm{C}_{1})}$}
\newcommand{\CTwo}{$\boldsymbol{(\mathrm{C}_{2})}$}
\newcommand{\CThree}{$\boldsymbol{(\mathrm{C}_{3})}$}
\newcommand{\Ffour}{$\boldsymbol{(\mathrm{C}_{4})}$}
\newcommand{\CBar}{$\boldsymbol{(\overline{\mathrm{C}}}\boldsymbol{)}$}

\begin{document}
\title{Rough multifactor volatility for SPX and VIX options}

\author{Antoine Jacquier}
\address{Department of Mathematics, Imperial College London and the Alan Turing Institute}
\email{a.jacquier@imperial.ac.uk}
\author{Aitor Muguruza}
\address{Department of Mathematics, Imperial College London and Kaiju Capital Management}
\email{aitor.muguruza-gonzalez15@imperial.ac.uk}
\author{Alexandre Pannier}
\address{Department of Mathematics, Imperial College London}
\email{a.pannier17@imperial.ac.uk}

\date{\today}

\thanks{AP acknowledges financial support from the EPSRC CDT in Financial Computing and Analytics.
AP and AJ are supported by an EPSRC EP/T032146/1 grant.}

\begin{abstract}
We provide explicit small-time formulae for the at-the-money implied volatility, skew and curvature 
in a large class of models, including rough volatility models and their multi-factor versions. 
Our general setup encompasses both European options on a stock and VIX options, 
thereby providing new insights on their joint calibration.
The tools used are essentially based on Malliavin calculus for Gaussian processes.
We develop a detailed theoretical and numerical analysis of the two-factor rough Bergomi model
and provide insights on the interplay between the different parameters for joint SPX-VIX smile calibration.
\end{abstract}

\maketitle

\tableofcontents

\section{Introduction}

Exposure to the uncertain dynamics of volatility is a desirable feature of most trading strategies and has naturally generated wide interest in volatility derivatives. 
From a theoretical viewpoint, an adequate financial model should reproduce the volatility dynamics accurately and consistently with those of the asset price; any dicrepancy may otherwise lead to arbitrage opportunity. 
Despite extensive research, implied volatility surfaces from options on the VIX and the S\&P 500 index 
still display discrepancies, betraying the lack of a proper modelling framework. 
This issue is well-known as a the \emph{SPX-VIX joint calibration problem} 
and has motivated a number of creative modelling innovations in the past fifteen years. 
Reconciling both implied volatilities requires additional factors to enrich the variance curve dynamics, 
as argued by Bergomi~\cite{Bergomi05,Bergomi08} 
where he proposed the multi-factor model
\begin{equation}\label{eq:Bergomi}
\frac{\D \xi_t(T)}{\xi_t(T)}=  \sum_{i=1}^{N}c_i \E^{-\kappa_i(T-t)} \D W_t^i,
\qquad 0\leq t\leq T,
\end{equation}
for the forward variance, with~$W^1,\ldots, W^N$ correlated Brownian motions
and~$c_1,\kappa_1,\ldots, c_N, \kappa_N>0$.
Gatheral~\cite{Gatheral08} recognised the importance of the additional factor 
to disentangle different aspects of the implied volatility 
and to allow humps in the variance curve, 
and introduced a mean-reverting version--double CEV model--where the instantaneous mean of the variance follows a CEV model itself. 
Although promising, these attempts fell short of reproducing jointly the short-time behaviour of the SPX and VIX implied volatilities.
A variety of new models were suggested to tackle this issue, both with continuous paths~\cite{BNP18,FS18,GIP17} and with jumps~\cite{BB14,CM14,CK13,KS15,PPR18,PS14}, incorporating novel ideas and increased complexity such as regime switching volatility dynamics.
Model-free bounds were also obtained in~\cite{DMHL15,GMN17,Guyon20,Papa14}, shedding light on the links between VIX and SPX and the difficulty of capturing them both simultaneously.

Getting rid of the restraining Markovian assumption that burdens classical stochastic volatility models has permitted the emergence of rough volatility models, 
which consistently agree with stylised facts under both the historical and the pricing measures~\cite{ALV07,BFG16,BLP16, Fukasawa11,GJR14}. 
A large portion of the toolbox developed for Markovian diffusion models is not available any longer
and asymptotic methods thus play a prominent role in understanding the theoretical properties of these models~\cite{GJRS18,HJL19,JPS18,JP20}.
Since the fit of the spot implied volatility skew is extremely accurate under this class of models~\cite{GJR14}, 
it seems reasonable to expect good results when calibrating VIX options. 
Moreover, the newly established hedging formula by Viens and Zhang~\cite{VZ18} shows that a rough volatility market is complete if it also contains a proxy of the volatility of the asset; this acts as an additional motivation for our work.
Still,~\cite{JMM17} showed that the rough Bergomi model is too close to lognormal to jointly calibrate both markets. Its younger sister~\cite{HJT18} added a stochastic volatility of volatility component, 
generating a smile sandwiched between the bid-ask prices when calibrating VIX, but the joint calibration is not provided. By incorporating a Zumbach effect, the quadratic rough Heston model~\cite{GJR20} achieves good results for the joint calibration at one given date. Further numerical methods were developed in~\cite{BCJ21,BDM21,RZ21}. However, the lack of analytical tractability of rough volatility models is holding back the progress of theoretical results on the VIX, with the notable exception of large devations results from~\cite{FGS21,LMS19} and the small-time asymptotics of~\cite{AGM18}.

In the latter, $\Ff_T$-measurable random variables (with volatility derivatives in mind) are written in the form of exponential martingales thanks to the Clark-Ocone formula, allowing the application of established asymptotic methods from~\cite{ALV07}. 
An expression for the short-time limit at-the-money (ATM) implied volatility skew is derived, yielding an analytical criterion that a model should satisfy to reproduce the correct short-time behaviour. The proposed mixed rough Bergomi model does meet the requirement of positive skew of the VIX implied volatility, backing its implementation with theoretical evidence. And indeed, the fits are rather satisfying. This model is built by replacing the exponential kernels of the Bergomi model~\eqref{eq:Bergomi} ($t\mapsto \E^{-\kappa t}$) with fractional kernels of the type~$t\mapsto t^{H-\half}$ with $H\in(0,\half)$, but is limited to a single factor, i.e. $W^1=W^2$. As a result, numerical computations under this model induce a linear smile, or equivalently a null curvature, unfortunately inconsistent with market observations.
To remedy this, we incorporate Bergomi's and Gatheral's insights on multi-factor models, 
(integrated by~\cite{DeMarco,LMS19} into rough volatility models) 
and extend~\cite{AGM18} to the multi-factor case;
we also compute the short-time ATM implied volatility curvature, deriving a second criterion for a more accurate model choice. 
In summary, the present paper goes beyond~\cite{AGM18} for three reasons:
\begin{itemize}
\item We consider multifactor models, far more efficient for VIX calibration, which complicate the analysis;
\item We compute the second derivative of the implied volatility to discriminate better between models; this turns out to be highly more technical than the skew;
\item We provide detailed proofs of all of our results at three levels: abstract model, generic rough volatility model for the VIX, and two-factor rough Bergomi model, checking carefully that all the assumptions are satisfied, proving technical lemmas applicable to our setting and exhibiting definite formulas at all three levels of generality.
\end{itemize}

We gather in Section~\ref{3sec:framework} our abstract framework and assumptions. 
The main results, short-time limits of the implied volatility level, skew, and curvature, 
are contained in Section~\ref{3sec:main}. 
Our framework covers a wide range of underlying assets, in particular VIX (Section~\ref{3sec:VIX}) 
and stock options (Section~\ref{3sec:spot}), 
in particular in Propositions~\ref{prop:genmodel} and~\ref{prop:SPXlimit}. 
We provide further a detailed analysis of the two-factor rough Bergomi model~\eqref{eq:Bergomi}.
Closed-form expressions that depend explicitly on the parameters of the model are provided in Proposition~\ref{prop:expomodel} for the VIX and Corollary~\ref{coro:SPXlimit} for the stock. 
They give insights on the interplay between the different parameters, and make the calibration task easier by allowing to fit some stylised facts prior to performing numerical computations. For instance, different combinations of parameters can yield positive or negative curvature.
All the proofs are gathered in the appendices, starting with useful lemmas and then following the order of the sections.

\vspace{0.2cm}
\textbf{Notations.}
For an integer~$N\in\NN$ and 
a vector $\xx\in\RR^N$, we denote $|\xx| := \sum_{i=1}^{N}x_i$ and
$\norm{\xx}^2 := \sum_{i=1}^{N}x_i^2$. 
We fix a finite time horizon~$T>0$ and let $\TT:=[0,T]$. For all~$p\ge1$, $L^p$ stands 
for the space~$L^p(\Omega)$
for some reference sample space~$\Omega$. 
As we consider rough volatility models, the Hurst parameter $H\in (0,\half)$ is a fundamental quantity
and we shall write $\Hp:=H+\half$ and $\Hm:=H-\half$.

\section{Framework}\label{3sec:framework}

We consider a square-integrable strictly positive process~$ (A_t)_{t\in\TT}$, adapted to the natural filtration~$ (\Ff_t)_{t\in\TT}$ 
of an $N$-dimensional Brownian motion $\Wb=(W^1,...,W^N)$ defined on a probability space 
$(\Omega,\Ff,\PP)$. 
We further introduce the true $(\Ff_t)_{t\in\TT}$-martingale conditional expectation process
$$
M_{t}:=\EE_t[A_T]:=\EE[A_T| \Ff_t], \qquad\text{for all  } t\in\TT.
$$

The set $\DD^{1,2}$ will denote the domain of the Malliavin derivative operator $\Df$ with respect to the Brownian motion $\Wb$, while~$\Df^i$ indicates the Malliavin derivative operator with respect to~$W^i$.
It is well known that $\DD^{1,2}$ is a dense subset of 
$L^{2}(\Omega)$ and that $\Df$ is a closed and unbounded operator from 
$L^{2}(\Omega)$ into $L^{2}(\TT\times\Omega)$. 
Analogously we define the sets of Malliavin differentiable processes~$\mathbb{L}^{n,2}:=L^{2}(\TT;\DD^{n,2})$. We refer to~\cite{Nualart06} for more details on Malliavin calculus.
Assuming~$A_T\in\DD^{1,2}$, the Clark-Ocone formula~\cite[Theorem 1.3.14]{Nualart06} reads,
for each $t\in\TT$,
\begin{equation}
\label{eq:MRT}
M_{t}
= \EE[M_{t}] + (\mb\bullet\Wb)_t
:= \EE[M_t] + \sum_{i=1}^N \int_0^t m^i_s \D W^i_s,
\end{equation}
where each component of~$\mb$ is
$m^i_s:=\EE\left[\Df^{i}_s A_T |\Ff_s\right]$.
Since~$M$ is a martingale, we may rewrite~\eqref{eq:MRT} as
\begin{equation}\label{eq:MRT2}
M_{t} = M_{0} + (M\phib\bullet \Wb)_t,
\end{equation}
where 
$\phib_{s} := \mb_s / M_s$ is defined whenever~$M_s\neq0$ almost surely. 
If $\phib=(\phi^1,...,\phi^N)$ belongs to~$\mathbb{L}^{n,2}$, then the following processes are well defined for all~$t<T$:
\begin{align}\label{eq:defprocesses}
Y_t&:=\int_t^T \norm{\phib_r}^2 \dr, \qquad \ut_{t} := \sqrt{Y_t}, \qquad u_{t} := \frac{\ut_t}{\sqrt{T-t}} ;\\
\Theta^i _{t}&:=\left(\displaystyle \int_{t}^{T}\Df^{i}_{t} \norm{\phib_r}^2 \D r\right)\phi^i_{t}, \qquad \mathrm{and} \qquad |\Thb| := \sum_{i=1}^n \Theta^i.
\end{align}
Note that all the processes depend implicitly on~$T$, 
which will be crucial when we study the short-time limit as~$T$ tends to zero.

\subsection{Level, skew and curvature}
Since~$M$ is a strictly positive martingale process, we can use it as an underlying to introduce options. A standard practice is to work with its logarithm~$\Mf :=\log(M)$, so that
$\Mf_T = \log\EE_T[A_T] = \log(A_T)$ and $\Mf_0 = \log\EE[A_T]$.
Under no-arbitrage arguments, the price~$\Pi_t$ at time~$t$ of a European Call option with maturity~$T$ and log-strike~$k\geq 0$ 
is equal to
\begin{equation}\label{eq:PriceVt}
\Pi_t(k) := \EE_t\left[\left(M_T-\E^k\right)^+\right] = \EE_t\left[\left(A_T-\E^k\right)^+\right],
\end{equation}
and the at-the-money value is denoted by $\Pi_t:=\Pi_t(\Mf_0) = \EE_t[(A_T - M_t)^+]$. 
We adapt the usual definitions of at-the-money implied volatility level, skew and curvature for the case where the underlying is a general process 
(later specified for the VIX and the S\&P).
Denote by $\BS(t,x,k,\sigma)$ the Black-Scholes price of a European Call option
at time $t\in\TT$, with maturity~$T$, log-stock~$x$, log-strike~$k$ and volatility~$\sigma$.
Its closed-form expression reads
\begin{equation}\label{eq:BSFormula}
\BS(t,x,k,\sigma )=
\left\{
\begin{array}{ll}
\E^{x}\Nn(d_{+}(x,k,\sigma ))-\E^{k}\Nn(d_{-}(x,k,\sigma )), & \text{if }\sigma\sqrt{T-t}>0,\\
\left(\E^x - \E^k\right)^+, & \text{if }\sigma\sqrt{T-t}=0,
\end{array}
\right.
\end{equation}
with
$d_{\pm }(x,k,\sigma) :=\frac{x-k}{\sigma \sqrt{T-t}}\pm 
\frac{\sigma \sqrt{T-t}}{2}$,
where $\Nn$ denotes the Gaussian cumulative distribution function.

\begin{definition}\label{def:ATMI skew}\ 
\begin{itemize}
\item For any $k\in\RR$, the implied volatility~$\Ii_{T}(k)$ is the unique non-negative solution to
$\Pi_0(k)=\BS\big(0,\Mf_0, k, \Ii_{T}(k)\big)$; 
we omit the $k$-dependence when considering it at-the-money ($k=\Mf_0$).    
\item The at-the-money implied skew~$\Ss$ 
and curvature~$\Cc$ at time zero are defined as
	$$
	\Ss_{T}:=\left|\pd_{k} \Ii_{T}(k)\right|_{k=\Mf_0}
	\qquad\text{and}\qquad
	\Cc_{T}:= \left|\pdk^2 \Ii_{T}(k)\right|_{k=\Mf_0}.
	$$ 
\end{itemize}
\end{definition}

\subsection{Examples}\label{subsec:examples}
The framework~\eqref{eq:MRT2} encompasses a large class of models, 
including stochastic volatility models ubiquitous in quantitative finance.
Consider a stock price process~$(S_t)_{t\in\TT}$, satisfying
$$
\frac{dS_t}{S_t} = \sqrt{v_t}\,\D B_t = \sqrt{v_t}\sum_{i=1}^N \rho_i\,\D W^i_t, 
$$
where~$v$ is a stochastic process adapted to~$(\Ff_t)_{t\in\TT}$, $\rhob:=(\rho_1,\cdots,\rho_N)\in[-1,1]^N$ 
with $\rhob\rhob^\top =1$. 

\subsubsection{Asset price}\label{ex:assetprice}
For $N=2$, the model~\eqref{eq:MRT2} corresponds to a one-dimensional stochastic volatility model by identifying~$A=M=S$, $\phi^1= \rho_1 \sqrt{v}$ and $\phi^2 = \rho_2 \sqrt{v}$, and $v$ is a process driven by $W^1$.
Our analysis generalises~\cite[Equation (2.1)]{ALV07} to the multi-factor case (in the continuous-path case). 
We refer to Section~\ref{3sec:spot} for the details in the multi-factor setting and the analysis of the implied volatility.

\subsubsection{VIX}\label{ex:VIX}
The VIX is defined as
$\VIX_{T}=\sqrt{\frac{1}{\Delta}\int_T^{T+\Delta}\EE_T[v_t ] \dt}$, where~$\Delta$ is one month.
The representation~\eqref{eq:MRT} yields that the underlying is the VIX future
$$
M_{t}^{\VIX}:=\EE_t[\VIX_T]=\EE[\VIX_T]+(\mb\bullet\Wb)_t,
\qquad\text{with}\qquad
m^i_s=\frac{1}{2\Delta}\EE_s\left[ \frac{1}{\VIX_T}\int_T^{T+\Delta}\Df^{i}_s v_{r} \dr\right].
$$
\subsubsection{Asian options}
For Asian options, the process of interest is 
$\mathcal{A}_T:=\frac{1}{T}\int_0^T S_t \dt$.
Using~\eqref{eq:MRT} we find
$$
M^{\mathcal{A}}_{t}:=\EE_t[\mathcal{A}_T]=\EE[\mathcal{A}_T] + (\mb\bullet\Wb)_t,
\qquad\text{with}\qquad
m^i_s=\frac{1}{T} \int_s^T \EE_s[\Df^{i}_s S_{r}] \dr.
$$
\subsubsection{Multi-factor rough Bergomi}\label{ex:multirB}
Rough volatility models can be written as $v_t=f(\Wb^H_t)$, where~$\Wb^H$ 
is an $N$-dimensional fractional Brownian motion with correlated components and~$f:\RR^N\to\RR$. For instance in the two-factor rough Bergomi model,
\[
v_t = v_0 \left( \chi \exp\left\{\nu W^{1,H}_t - \frac{\nu^2}{2} \EE\left[\left(W^{1,H}_t\right)^2\right]\right\}
+(1- \chi) \exp\left\{\eta W^{2,H}_t - \frac{\eta^2}{2} \EE\left[\left(W^{2,H}_t\right)^2\right]\right\}\right),
\]
with $\chi \in (0,1)$, $\nu, \eta, v_0>0$.
In Example~\ref{ex:VIX} we set $A=\VIX$ and hence $N=2$, but in the asset price case we set~$A=S$ and therefore $N=3$ even though the variance only depends on two factors.

\subsection{General assumptions}
We introduce the following broad assumptions, key for the whole analysis, and provide in Section~\ref{3sec:VIX} sufficient conditions to simplify them in the VIX case.
\begin{description}
	\item[\HOne] \label{H1}  $A\in \mathbb{L}^{4,p}$. 
	\item[\HTwo] \label{H2} $\displaystyle\frac{1}{M_t}\in L^p$, for all $p>1$, and all~$t\in\TT$.
	\item[\HThree] \label{H3} The term 
	$\displaystyle \EE_t\left[ \int_{t}^{T}\frac{|\Thb_{s}|}{\ut_{s}^2}\ds\right]$ 
	is well defined for all~$t\in\TT$.
	\item[\HFour] \label{H4} The term~$\displaystyle \frac{1}{\sqrt{T}} \EE\left[ \int_0^T \frac{|\Thb_s|}{\ut_s^2} \ds\right]$ tends to zero as~$T$ tends to zero. 
	
	\item[\HFive] \label{H5} There exists~$p\ge1$ such that~$\sup_{T\in[0,1]}\ut_0^p<\infty$ almost surely and, for all random variables~$Z\in L^p$ and all $i\in\llbracket1,N\rrbracket$, 
	the following terms are well defined and tend to zero as $T$ tends to zero:
	$$
	\int_0^T \EE\left[Z \left( \EE_s \left[ \frac{1}{u_0} \int_0^T \Df^i_s \norm{\phib_r}^2 \dr \right]\right)^2\right] \ds.
	$$
\end{description}
There exists~$\lambda\in(-\half,0]$ such that:
\begin{description}
	\item[\HSix] These expressions converge to zero as~$T$ tends to zero:
	$$
	\frac{1}{T^{\half+\lambda}} \EE\left[\int_0^T \frac{|\Thb_s|\int_s^T|\Thb_r|\dr}{\ut_{s}^{6}}\ds\right] \quad \text{and} \quad
	\frac{1}{T^{\half+\lambda}}\EE\left[\int_0^T \frac{1}{\ut_{s}^{4}} \sum_{k=1}^N \left\{ \phi^k_s  \Df^k_s \left( \int_s^T |\Thb_r| \dr \right) \right\} \ds\right].$$
	\item[\HSeven] The random variable~$\Kf_T:=\displaystyle \frac{\int_0^T |\Thb_s|\ds}{T^{\half+\lambda} \ut_{0}^{3}}$ is such that $\EE[\ut_0^2 \Kf_T]$ tends to zero and~$\EE[\Kf_T]$ has a finite limit as~$T$ tends to zero.
\end{description}
There exists~$\gamma\in(-1,0]$ such that:
\begin{description}
	\item[\HEight]  The following expressions converge to zero as~$T$ tends to zero:
	\begin{align*}
	& \frac{1}{T^{\half+\gamma}}\EE \left[ \int_0^T \ut_{s}^{-10} |\Thb_s|
	\left( \int_s^T |\Thb_r|\left( \int_r^T |\Thb_{y}|\dy \right) \dr \right) \ds
	\right], \\ 
	& \frac{1}{T^{\half+\gamma}} \EE \left[ \int_0^T \ut_{s}^{-8} \sum_{j=1}^N \left( \phi^j_s \Df^j_s \left( \int_s^T |\Thb_r| \left( \int_r^T |\Thb_{y}|\dy \right) \dr \right) \ds \right) \ds \right], \\ 
	& \frac{1}{T^{\half+\gamma}} \EE \left[ \int_0^T \ut_{s}^{-8}|\Thb_s|
	\int_s^T \sum_{j=1}^N \left\{ \phi^j_r  \Df^j_r \left( \int_r^T |\Thb_{y}|\dy \right) \right\} \dr \ds \right], \\  
	& \frac{1}{T^{\half+\gamma}} \EE \left[ \int_0^T \ut_{s}^{-6 } \sum_{k=1}^N \left\{ \phi^k_s \Df^k_s \left( \int_s^T \sum_{j=1}^N \left\{ \phi^j_r \, \Df^j_r \left( \int_r^T |\Thb_{y}| \dy \right) \right\} \dr \right) \right\} \ds \right]. 
	\end{align*}
	\item[\HNine]  The random variables
	$$
	\Hf_T^1 := \frac{1}{T^{\half+\gamma} \ut_{0}^{7}}\int_0^T |\Thb_s| \left(\int_s^T |\Thb_r|\dr\right)\ds
	\quad\text{and}\quad
	\Hf_T^2 := \frac{1}{T^{\half+\gamma} \ut_{0}^{5}} \int_0^T \sum_{j=1}^N \left\{ \phi^j_s \Df^j_s \left( \int_s^T |\Thb_{r}| \dr \right) \right\} \ds,
	$$
	are such that~$\EE\big[(\ut_0^6+\ut_0^4+\ut_0^2) \Hf^1_T + (\ut_0^4+\ut_0^2) \Hf_T^2 \big]$ tend to zero and both~$\EE[\Hf^1_T]$ and~$\EE[\Hf^2_T]$ have a finite limit as~$T$ tends to zero.
\end{description}
\begin{remark}\label{rem:ConditionsH}\
\begin{itemize}
\item  {\HOne} requires~$A$ to be four times Malliavin differentiable. 
This is necessary to prove the curvature formula
using the Clark-Ocone formula~\eqref{eq:MRT} and three times the anticipative It\^o formula.
\item 
When the underlying is the stock price (as in Section~\ref{ex:assetprice}), it satisfies equation~\eqref{eq:MRT2} where~$\phi$ corresponds to its volatility $\sqrt{v}$. One can then directly make assumptions on the variance process, as in~\cite{AL17,ALV07,AS19}. We make this explicit in Proposition~\ref{prop:SPXlimit} for example.
In the case of the VIX (Section~\ref{sec:GenVolModel}), $\phi$ is much more intricate which is why we refrain ourselves from doing the same. Nevertheless, sufficient conditions are given by {\CBar}.
\end{itemize}
\end{remark}


\section{Main results} \label{3sec:main}
We gather here our main asymptotic results for our general framework above,
with the proofs postponed to Appendix~\ref{app:proofsmain} to ease the flow. 
The first theorem states that the small-time limit of the implied volatility is equal to the limit of the forward volatility. 
This is well known for Markovian stochastic volatility models in~\cite{AS19,BBF04} and in a one-factor setting~\cite{AGM18}.
To streamline the call to the assumptions, we shall group them using mixed subscript notations, for example {\HOT} corresponds to {\HOne-\HTwo-\HThree}
and we further write {\HBarl} to mean {\HAll-\Hl} and {\HBarlg} as short for \HAll-\Hl-\Hg.

\begin{theorem}
	\label{thm:level}
	If {\HAll} hold,
	then
	$$
	\lim_{T\downarrow 0} \Big(\Ii_{T} - \EE[u_{0}]\Big)=0.
	$$
\end{theorem}
Note that we did not assume the limit of~$\EE[u_0]$ to be finite.
The proof, in Appendix~\ref{sec:prooflevel}, builds on arguments from ~\cite[Proposition 3.1]{AS19}. 
We then turn our attention to the ATM skew, defined in~\ref{def:ATMI skew}. This short-time asymptotic is reminiscent of~\cite[Proposition 6.2]{ALV07} and~\cite[Theorem 8]{AGM18}. 
\begin{theorem} \label{thm:skew}
	If there exists~$\lambda\in(-\half,0]$ such that {\HBarl} are satisfied, then 
	\begin{equation}\label{eq:skewmain}
	\lim_{T\downarrow 0} \frac{\Ss_T}{T^{\lambda}}
	= \half\lim_{T\downarrow 0 } \EE\left[\frac{1}{T^{\half+\lambda}}\frac{ \int_0^T|\Thb_s|\ds}{\ut_{0}^3}\right].
	\end{equation}
\end{theorem}
Note that~\eqref{eq:skewmain} still holds without \HSeven, but in that case both sides are infinite.
In the rough volatility setting of Section~\ref{ex:assetprice} with~$v_t=f(\Wb^H_t)$, $\lambda$ corresponds to~$H-\half$ such that~\eqref{eq:skewmain} matches the slope of the observed ATM skew of SPX implied volatility. We prove this theorem in Appendix~\ref{sec:proofskew}.
We also provide the short-term curvature, proved in Appendix~\ref{sec:proofcurvature}:
\begin{theorem}\label{thm:curvature}
	If there exist~$\lambda\in(-\half,0]$ and~$\gamma\in(-1,\lambda]$ ensuring {\HBarlg}, then
	\begin{align}
	\lim_{T\downarrow 0}\frac{\Cc_T}{T^{\gamma}}
	= \lim_{T\downarrow 0} \frac{1}{T^{\gamma}} \Bigg\{& \Ss_T 
	-\frac{15}{2\sqrt{T}} \EE \left[ \frac{1}{\ut_0^7} \int_0^T |\Thb_r|\left( \int_r^T |\Thb_{y}|\dy \right) \dr \right] \nonumber\\
	&+ \frac{3}{2\sqrt{T}}\EE \left[ \frac{1}{\ut_0^5} \int_0^T \sum_{j=1}^N \left\{ \phi^j_s \Df^j_s \left( \int_s^T |\Thb_{y}|\dy \right) \right\} \ds \right]
	\Bigg\}.
	\label{eq:curvature}
	\end{align}
	The limit still holds without \HNine\, but in that case the second and third term are infinite.
\end{theorem}

Note that \HSeven\, with~$\lambda\ge\gamma$ guarantees that the first term $T^{-\gamma} \Ss_T$ converges. By Theorem~\ref{thm:skew}, 
	\begin{equation*}
	\lim_{T\downarrow0} \frac{\Ss_T}{T^{\gamma}} =
	\left\{ 
	\begin{array}{ll}
	0, \quad & \text{if } \lambda>\gamma,\\
	\displaystyle\half\lim_{T\downarrow 0 } \EE\left[\frac{1}{T^{\half+\lambda}}\frac{ \int_0^T|\Thb_s|\ds}{\ut_{0}^3}\right], & \text{if } \lambda=\gamma, \\
	+\infty, & \text{if } \lambda<\gamma.
	\end{array}
	\right.
	\end{equation*}
	
\section{Asymptotic results in the VIX case}\label{3sec:VIX}
As advertised, our framework includes the VIX case where
$$ A_T =\VIX_T=\sqrt{\frac{1}{\Delta}\int_T^{T+\Delta}\EE_T[v_r] \dr},$$
for~$v_r\in\DD^{3,2}$ for all~$r\in[0,T+\Delta]$ and we provide simple sufficient conditions for {\HBarlg} to hold. 

\subsection{A generic volatility model}\label{sec:GenVolModel}
Consider the following four conditions which we gather under the notation~\CBar:
there exist $H \in (0,\half)$ and~$X\in L^p$ for all~$p>1$ such that 
\begin{enumerate}
	\item[\COne] For all~$t\ge0$, $\frac{1}{M_t} \le X$ almost surely;
	\item[\CTwo] For all~$i,j,k\in\llbracket1, N\rrbracket$ and~$t\le s\le y\le T \le r$, we have, almost surely
	\begin{itemize} 
		\item $v_r\le X$,
		\item $\Df^i_y v_r \le X (r-y)^{\Hm}$, 
		\item $\Df^j_s \Df^i_y v_r \le X (r-s)^{\Hm} (r-y)^{\Hm}$,
		\item $\Df^k_t \Df^j_s \Df^i_y v_r \le X (r-t)^{\Hm} (r-s)^{\Hm} (r-y)^{\Hm}$;
	\end{itemize}
	\item[\CThree] For all~$p>1$, $\EE[u_s^{-p}]$ is uniformly bounded in~$s$ and~$T$, with~$s\le T$.
	\item[\Ffour]
	For all~$ i,j,k\in\llbracket1, N\rrbracket$ and~$r\ge0$, the mappings~$y\mapsto\Df^i_y v_r$, $s\mapsto \Df_s^j \Df^i_y v_r$, and~$t\mapsto \Df^k_t \Df^j_s \Df^i_y v_r$ are almost surely continuous in a neighbourhood zero.
\end{enumerate}

Recall the notations~$\Hm$ and~$\Hp$ from the introduction.
We compute the level, skew, and curvature of the VIX implied volatility in a model which satisfies the sufficient conditions. 
Let us define the following limits
\begin{equation}\label{eq:defJG}
J_i := \int_0^\Delta \EE[\Df^i_0 v_r] \dr ,\qquad G_{ij} := \int_0^\Delta \EE\big[\Df^j_0 \Df^i_0 v_r \big] \dr, \quad \text{for all } i,j\in\llbracket1,N\rrbracket. 
\end{equation}
\begin{proposition}\label{prop:genmodel}
Under \CBar, the following limits hold:
\begin{equation*}
\begin{array}{rll}
\displaystyle\lim_{T\downarrow 0} \Ii_{T} & = \displaystyle\frac{\norm{\Jb}}{2\Delta\VIX_0^2}, & \text{if }H\in\left(0,\half\right),\\
\displaystyle\lim_{T\downarrow 0} \Ss_T & =  \displaystyle\frac{\sum_{i,j=1}^N J_i J_j \left(G_{ij} - \frac{J_i J_j}{\Delta\VIX_0^2}\right)}{2\norm{\Jb}^3},  & \text{if }H\in\left(0,\half\right),\\
\displaystyle\lim_{T\downarrow 0} \frac{\Cc_T}{T^{3H-\half}}
	& = \displaystyle\frac{2\Delta \VIX_0^{2}}{3\norm{\Jb}^5} \sum_{i,j,k=1}^N J^i J^j J^k \,\lim_{T\dto0} \frac{\int_T^{T+\Delta}  \EE\left[\Df^k_0 \Df^j_0\Df^i_0 v_r\right]\dr}{T^{3H-\half}}, & \text{if }H\in\left(0,\frac{1}{6}\right).
\end{array}
\end{equation*}
\end{proposition}
\begin{remark}
    Our results stand under the fairly general assumption $\bm{(\overline{C})}$. If $v$ is a reasonably well-behaved function of an $N$-dimensional Gaussian Volterra process $(W^{1,H},\cdots, W^{N,H})$ then these should be relatively easy to check, as Proposition~\ref{prop:suffconds} suggests. The rough Heston model is not even known to be Malliavin differentiable to this day, thus it does not lie in the scope of this study.
\end{remark}
We split the proof in two steps, collected in Appendix~\ref{sec:VIXproof}. 
First we show that \COne, \CTwo, \CThree\, are sufficient to apply our main theorems as they imply {\HBarlg} for any
$\lambda\in(-\half,0]$ and~$\gamma\in(-1, 3H-\half]$. Thanks to {\Ffour} we can also compute the limits---after a rigorous statement of convergence results---starting with~$\Ii_T$ and the skew with~$\lambda=0$. 
Restricting~$H$ to $(0,1/6)$, which is  the most relevant regime for rough volatility models, we can set~$\gamma=3H-\half<\lambda$ and compute the short-time curvature, with only the second term in {\HNine} contributing to the limit. The curvature limit in Proposition~\ref{prop:genmodel} is finite by the last point of~\CTwo.

\begin{remark}\label{rem:Hhigher}
In the regime $H\in[1/6,1/2)$, the rescaling becomes $\gamma=0$ and many more terms that would just vanish when $H<1/6$ now give a non-trivial contribution in the limit. Informally (that is without a proof), the limit reads
\begin{align*}
\lim_{T\downarrow0} \Cc_T 
=& \lim_{T\downarrow0}\Ss_T -  \frac{15\Delta\VIX^2_0}{\norm{\Jb}^7}
\left(\sum_{i,j=1}^N J_i J_j \left(G_{ij} - \frac{J_i J_j}{\Delta\VIX_0^2}\right)\right)^2 \\
& + \frac{12\Delta \VIX_0^2}{\norm{\Jb}^5}\sum_{i,j,k=1}^N  J_i \left( G_{jk}-\frac{J_j J_k}{\Delta \VIX_0^2}\right) \left( G_{ik}-\frac{J_i J_k}{\Delta \VIX_0^2}\right) \\ 
& + \frac{1}{\norm{\Jb}^5} \sum_{i,j,k=1}^N \left\{ \frac{9 \big( J_i J_j J_k \big)^2}{2 \Delta \VIX_0^2} -  6 J_i J_j J_k \left(G_{ij} J_k + G_{ik} J_j + G_{jk} J_i \right) \right\} \\
& + \frac{2\Delta \VIX_0^{2}}{3\norm{\Jb}^5} \sum_{i,j,k=1}^N J^i J^j J^k \,\int_0^{\Delta}  \EE\left[\Df^k_0 \Df^j_0\Df^i_0 v_r\right]\dr.
\end{align*}
\end{remark}


\subsection{The two-factor rough Bergomi}
We consider the two-factor exponential model
\begin{equation}\label{eq:expomodel}
v_t = v_0 \left[ \chi \Ee\left(\nu W^{1,H}_t\right) + \chib \Ee\left(\eta\left(\rho W^{1,H}_t + \rrho W^{2,H}_t\right)\right) \right]=: v_0\big( \chi\Ee_t^1 + \chib \Ee_t^2\big),
\end{equation}
where~$H\in(0,\half]$, $W^{i,H}_t= \int_0^t (t-s)^{\Hm}  \D W^i_s$, $W^1,W^2$ are independent Brownian motions, the Wick exponential is defined as~$\Ee(X):=\exp\{X-\half \EE[X^2]\}$, for any random variable~$X$, and~$\chi\in[0,1]$, $\chib:=1-\chi$, $v_0,\nu,\eta>0$, $\rho\in[-1,1]$, $\rrho=\sqrt{1-\rho^2}$.
This model is an extension of the Bergomi model~\cite{Bergomi08}, 
where the exponential kernel is replaced by a fractional one
and an extension of the rough Bergomi model~\cite{BFG16} to the two-factor case. 
It combines Bergomi's insights on the need for several factors with the benefits of rough volatility.
As proved in Appendix~\ref{app:proofconds}, it satisfies our conditions:
\begin{proposition}\label{prop:suffconds}
	If~$\rho\in(-\sqrt{2}/2,1]$, the model~\eqref{eq:expomodel} satisfies~{\CBar}.
\end{proposition}
The restriction of the range of~$\rho$ is equivalent to~$\rho+\rrho>0$, a necessary requirement in the proof.
Proposition~\ref{prop:genmodel} therefore applies and we obtain the following limits,
with proof in Appendix~\ref{sec:proofexpomodel}:
\begin{proposition} \label{prop:expomodel}
Let~$\psi(\rho,\nu,\eta,\chi):=\sqrt{ (\chi\nu+\chib\eta\rho)^2+ \chib^2\eta^2\rrho^2}$. 
	If~$H\in(0,\frac{1}{6})$ and~$\rho\in(-\frac{\sqrt{2}}{2},1]$, then
	\begin{align*}
	\lim_{T\downarrow 0} \Ii_{T} & = \frac{\Delta^{\Hm}}{2H+1} \psi(\rho,\nu,\eta,\chi), \nonumber\\
	\lim_{T\downarrow0} \Ss_T & = \frac{\Hp\Delta^{\Hm}}{2\psi(\rho,\nu,\eta,\chi)^{3}} \Bigg\{(\chi\nu+\chib\eta\rho)^2 \left[ \frac{\chi\nu^2+\chib\eta^2\rho^2}{2H}-\left(\frac{\chi\nu+\chib\eta\rho}{\Hp}\right)^2\right]\\
	& \quad + 2 (\chi\nu+\chib\eta\rho)\chib^2\eta^2\rrho^2 \left[ \frac{\eta\rho}{2H} - \frac{\nu+\eta\rho}{\Hp^2}\right] 
	+  \chib^3 \eta^4\rrho^4 \left(\frac{1}{2H}-\frac{1}{\Hp^2}\right) \Bigg\},\\
	\lim_{T\downarrow 0} \frac{\Cc_T}{T^{3H-\half}}
	& = \frac{128 \Delta^{-2H} \Hp^2}{3\psi(\rho,\nu,\eta,\chi)^{5}(1-6H)}
	\Big\{ (\chi\nu+\chib\eta\rho)^3 (\chi\nu^3+\chib\eta^3\rho^3) \\
	&\quad + 3 (\chi\nu+\chib\eta\rho)^2 \chib^2 \eta^4\rrho^2 \rho^2
	+ 3  (\chi\nu+\chib\eta\rho)\chib^3  \eta^5\rrho^4 \rho 
	+  \chib^4\eta^6\rrho^6 \Big\}.
	\end{align*}
\end{proposition}
The limits depend explicitly on the parameters of the model $(H,\chi,\nu,\eta,\rho)$ and can be used to gain insight on their impact over the quantities of interest.

\begin{remark}\ 
\begin{itemize}
\item In the case~$\rho=1$ (hence~$\rrho=0$) the above limits simplify to
	\begin{align*}
	\lim_{T\downarrow 0} \Ii_{T} & = \frac{\Delta^{\Hm}}{2\Hp}\left(\chi\nu+\chib\eta\right),\\ 	\lim_{T\downarrow 0} \Ss_T & =\half \frac{\Hp\Delta^{\Hm}}{\chi\nu+\chib\eta} \left[\frac{\chi\nu^2+\chib\eta^2}{2H}-\left(\frac{\chi\nu+\chib\eta}{\Hp}\right)^2\right],\\
\lim_{T\downarrow 0} \frac{\Cc_T}{T^{3H-\half}}
	& = \frac{128 \Delta^{-2H} \Hp^2}{3-18H} \frac{\chi\nu^3+\chib\eta^3}{(\chi\nu+\chib\eta)^2}.
	\end{align*}
\item If we set $\rho=0$ (hence $\rrho=1$), we obtain
	\begin{align*}
	\lim_{T\downarrow 0} \Ii_{T} & = \frac{\Delta^{\Hm}}{2\Hp}\sqrt{\chi^2\nu^2+\chib^2\eta^2},\\
	\lim_{T\downarrow0} \Ss_T & = \frac{\Hp\Delta^{\Hm}}{2(\chi^2\nu^2+\chib^2\eta^2)^{3/2}} \Bigg\{\chi^3\nu^4 \left[ \frac{1}{2H}-\frac{\chi}{\Hp^2}\right]
- \frac{2 \chi\nu^2\chib^2\eta^2}{\Hp^2}
	+  \chib^3 \eta^4\left(\frac{1}{2H}-\frac{1}{\Hp^2}\right) \Bigg\},\\
	\lim_{T\downarrow 0} \frac{\Cc_T}{T^{3H-\half}}
	& = \frac{128 \Delta^{-2H} \Hp^2}{3-18H} \frac{\chi^4\nu^6+\chib^4\eta^6}{(\chi^2\nu^2+\chib^2\eta^2)^{5/2}}.
	\end{align*}
\item When~$\rho=-1$ (not covered per se by the proposition) the above limits simplify to
	\begin{align*}
 \lim_{T\downarrow0} \mathcal{I}_T &= \frac{\Delta^{\Hm}}{2H+1} \lvert\chi\nu-\chib\eta\lvert,\\
  \lim_{T\downarrow0} \mathcal{S}_T &= \frac{H_{+}\Delta^{\Hm}}{2\lvert\chi\nu-\chib\eta\lvert} \left[ \frac{\chi\nu^2+\chib\eta^2}{2H} - \left(\frac{\chi\nu-\chib\eta}{\Hp}\right)^2\right],
  \\
\lim_{T\downarrow 0} \frac{\Cc_T}{T^{3H-\half}}
	& = \frac{128 \Delta^{-2H} \Hp^2}{3-18H} \frac{\chi\nu^3-\chib\eta^3}{(\chi\nu-\chib\eta)^2} \,{\rm sgn}(\chi\nu-\bar{\chi}\eta).
	\end{align*}
\end{itemize}
\end{remark}
Some tedious yet straightforward manipulations allow us to obtain some information about the sign of the limiting curvature:
\begin{lemma}
For any $\eta,\nu>0$, $\chi \in [0,1]$, there exists $\rho^*_{\chi,\nu\,\eta}<0$ such that $\lim_{T\downarrow 0} \frac{\Cc_T}{T^{3H-\half}}$ is strictly positive for $\rho>\rho^*_{\chi,\nu\,\eta}$ and strictly negative when 
$\rho<\rho^*_{\chi,\nu\,\eta}$.
When $\left(\chi\nu - \chib\eta\right)\left(\chi\nu^3 - \chib\eta^3\right)>0$, then
$\rho^*_{\chi,\nu\,\eta}<-1$ and hence
the limiting curvature is strictly positive for all $\rho \in [-1,1]$.
\end{lemma}
\begin{proof}
The expression we are interested in, given in Proposition~\ref{prop:expomodel}, and ignoring the obviously strictly positive multiplicative factor, reads
{\small
\begin{align*}
\Phi(\rho)
 & = {\color{blue} [\chi\nu+\chib\eta\rho]^3
 \left(\chi\nu^3+\chib\eta^3\rho^3\right)}
{\color{red} + 3\{\chi\nu+\chib\eta\rho\}^2\chib^2\eta^4\rhob^2\rho^2}
 {\color{olive} + 3(\chi\nu+\chib\eta\rho)\chib^3\eta^5\rhob^4\rho}
{\color{orange} + \chib^4\eta^6\rhob^6}\\
 & = {\color{blue} \Big[\chi^3\nu^3 + 3\chi^2\nu^2 \chib\eta\rho + 3\chi\nu\chib^2\eta^2\rho^2 + \chib^3\eta^3\rho^3\Big]
\left(\chi\nu^3+\chib\eta^3\rho^3\right)}
{\color{red} + 3\Big\{\chi^2\nu^2 + 2\chi\nu\chib\eta\rho +\chib^2\eta^2\rho^2\Big\}\chib^2\eta^4\rhob^2\rho^2}\\
 & \quad {\color{olive} + 3(\chi\nu+\chib\eta\rho)\chib^3\eta^5\rhob^4\rho}
{\color{orange} + \chib^4\eta^6\rhob^6}\\
 & = {\color{blue} \Big[\chi^3\nu^3 + 3\chi^2\nu^2 \chib\eta\rho + 3\chi\nu\chib^2\eta^2\rho^2 + \chib^3\eta^3\rho^3\Big]
\left(\chi\nu^3+\chib\eta^3\rho^3\right)}
{\color{red} + 3\Big\{\chi^2\nu^2 + 2\chi\nu\chib\eta\rho +\chib^2\eta^2\rho^2\Big\}\chib^2\eta^4\left(1-\rho^2\right)\rho^2}\\
 &\quad  {\color{olive} + 3(\chi\nu+\chib\eta\rho)\chib^3\eta^5\left(1-2\rho^2+\rho^4\right)\rho}
{\color{orange} + \chib^4\eta^6
\left(1 - 3\rho^2 + 3\rho^4 - \rho^6\right)}\\
 & =: \sum_{i=0}^{6}\alpha_i \rho^i,
\end{align*}
}
where
\begin{equation*}
\begin{array}{rll}
\alpha_0 & = {\color{blue} \chi^4\nu^6} {\color{orange} + \chib^4\eta^6}
 & = \chi^4\nu^6 + \chib^4\eta^6,\\
\alpha_1 & = {\color{blue} 3\chi^3\nu^5\chib\eta} {\color{olive} +3\chi\nu\chib^3\eta^5}
 & = 3\left(\chi^2\nu^4 + \chib^2\eta^4\right)\chi\nu\eta\chib,\\
\alpha_2 & = {\color{blue} 3\chi^2\chib^2\eta^2\nu^4} {\color{red} +3\chi^2\nu^2\chib^2\eta^4} {\color{olive} +3\chib^4\eta^6} 
{\color{orange} -3\chib^4\eta^6}
 = {\color{blue} 3\chi^2\chib^2\eta^2\nu^4} {\color{red} +3\chi^2\nu^2\chib^2\eta^4}
 & = 3\Big(\nu^2 + \eta^2\Big)\chi^2\chib^2\eta^2\nu^2\\
\alpha_3 & = {\color{blue} \chi^3\nu^3\chib\eta^3+\chib^3\eta^3\chi\nu^3} {\color{red} +6\chi\nu\chib^3\eta^5} 
{\color{olive} -6\chi\nu\chib^3\eta^5}
= {\color{blue} \left(\chi^2+\chib^2\right)\chi\chib\nu^3\eta^3}
& = \left(\chi^2+\chib^2\right)\chi\chib\nu^3\eta^3,\\
\alpha_4 & = {\color{blue} 3\chi^2\nu^2 \chib^2\eta^4} {\color{red} -3\chi^2\nu^2\chib^2\eta^4 + 3\chib^4\eta^6} 
{\color{olive} -6\chib^4\eta^6} {\color{orange} +3\chib^4\eta^6} & = 0,\\
\alpha_5 & = {\color{blue} 3\chib^3\eta^5\chi\nu} {\color{red} - 6\chib^3\eta^5\chi\nu} {\color{olive} + 3\chi\nu\chib^3\eta^5} & = 0,\\
\alpha_6 & =  {\color{blue} \chib^4\eta^6}  {\color{red}- 3\chib^4\eta^6} {\color{olive} +  3\chib^4\eta^6}
-\chib^4\eta^6 & = 0.
\\
\end{array}
\end{equation*}
These surprising simplifications show that~$\Phi$ is in fact a polynomial in~$\rho$ of order three, 
with a strictly positive leading coefficient~$\alpha_3$,
so that~$\Phi''$ is linear and increasing in~$\rho$ and is such that
$$
\Phi''(\rho) = 0
\quad\text{if and only if}\quad
\rho = -\frac{(\eta^2+\nu^2)\chib\chi}{\nu\eta(\chi^2+\chib^2)}
=:\rho_{\bullet}(\chi, \nu, \eta).
$$
Now, 
\begin{align*}
 & \Phi'(\rho_{\bullet}(\chi, \nu, \eta))\\
& = 3\alpha_3\rho^*(\eta,\nu\chi)^2 + 2 \alpha_2\rho^*(\eta,\nu\chi) + \alpha_1\\
& = 3\left(\chi^2+\chib^2\right)\chi\chib\nu^3\eta^3\rho^*(\eta,\nu\chi)^2
 + 6\left(\nu^2 + \eta^2\right)\chi^2\chib^2\eta^2\nu^2\rho^*(\eta,\nu\chi)
 + 3\left(\chi^2\nu^4 + \chib^2\eta^4\right)\chi\nu\eta\chib
\\
& = 3\left(\chi^2+\chib^2\right)\chi\chib\nu^3\eta^3\left[\frac{(\eta^2+\nu^2)\chib\chi}{\nu\eta(\chi^2+\chib^2)}\right]^2
 - 6\left(\nu^2 + \eta^2\right)\chi^2\chib^2\eta^2\nu^2\frac{(\eta^2+\nu^2)\chib\chi}{\nu\eta(\chi^2+\chib^2)}
 + 3\left(\chi^2\nu^4 + \chib^2\eta^4\right)\chi\nu\eta\chib
\\
& = 3\chi^3\chib^3\nu\eta\frac{(\eta^2+\nu^2)^2}{\chi^2+\chib^2}
 - 6\chi^3\chib^3\eta\nu\frac{(\eta^2+\nu^2)^2}{\chi^2+\chib^2}
 + 3\left(\chi^2\nu^4 + \chib^2\eta^4\right)\chi\nu\eta\chib
\\
& = 
\frac{1}{\chi^2+\chib^2}
\Big(-3\chi^3\chib^3\nu\eta(\eta^2+\nu^2)^2
 + 3(\chi^2+\chib^2)\left(\chi^2\nu^4 + \chib^2\eta^4\right)\chi\nu\eta\chib
\Big)\\
& = 
\frac{3\nu\eta\chi\chib}{\chi^2+\chib^2}
\Big((\chi^2+\chib^2)\left(\chi^2\nu^4 + \chib^2\eta^4\right)
-\chi^2\chib^2(\eta^2+\nu^2)^2
\Big)\\
& = 
\frac{3\nu\eta\chi\chib}{\chi^2+\chib^2}
\Big(
\chi^4\nu^4 + \chi^2\chib^2\eta^4
+ \chi^2\chib^2\nu^4 + \chib^4\eta^4
-\chi^2\chib^2\left(\eta^4+2\eta^2\nu^2+\nu^4\right)
\Big)\\
& = 
\frac{3\nu\eta\chi\chib}{\chi^2+\chib^2}
\Big(
\chi^4\nu^4 + \chib^4\eta^4
-2\chi^2\chib^2\eta^2\nu^2
\Big)
 = 
\frac{3\nu\eta\chi\chib}{\chi^2+\chib^2}
\left(\chi\nu + \chib\eta\right)^2
\left(\chi\nu - \chib\eta\right)^2>0.
\end{align*}
Since~$\Phi'$ is an upward parabola with strictly positive minimum, 
it is always strictly positive, hence~$\Phi$ is a strictly increasing function (of~$\rho$), and the lemma follows.
Let $\rho^{*}_{\chi, \nu, \eta}$ denote the unique
solution to $\Phi(\rho^{*}_{\chi, \nu, \eta})=0$,
for which an explicit solution exists in closed form, but its exact representation is messy and not particularly informative.
We can however provide an upper bound. Indeed
\begin{align*}
 \Phi(-1) 
 &  = \alpha_0 - \alpha_1 + \alpha_2 - \alpha_3\\
 & = \chi^4\nu^6 + \chib^4\eta^6
 - 3\left(\chi^2\nu^4 + \chib^2\eta^4\right)\chi\nu\eta\chib
 + 3\left(\nu^2 + \eta^2\right)\chi^2\chib^2\eta^2\nu^2
 -\left(\chi^2+\chib^2\right)\chi\chib\nu^3\eta^3\\
 & = \Big(\chi^3\nu^3
 - 3\chi^2\nu^2\eta\chib
 + 3\chi\chib^2\eta^2\nu
 -\chib^3\eta^3\Big)\chi\nu^3
  + \Big(\chib^3\eta^3 - 3\chib^2\eta^2\chi\nu
   + 3\chi^2\chib\eta\nu^2
   -\chi^3\nu^3\Big)\chib\eta^3\\
& = \left(\chi\nu - \chib\eta\right)^3\left(\chi\nu^3 - \chib\eta^3\right)
 \end{align*}
 As soon as $\Phi(-1)>0$, then clearly
 $\rho^{*}_{\chi, \nu, \eta}<-1$, hence the limiting curvature is always strictly positive.
 The sign of $\Phi(-1)$ is given by that of 
 $\left(\chi\nu - \chib\eta\right)\left(\chi\nu^3 - \chib\eta^3\right)$, 
 which is an upward parabola in~$\chi$.
\end{proof}

\section{The stock smile under multi-factor models}\label{3sec:spot}
We use the setting of Section~\ref{ex:assetprice} to apply our results to an asset price of the form
\[
S_t = S_0 + \int_0^t S_r \sqrt{v_r} \,\D B_r,
\]
where~$B$ is correlated with the other~$N$ Brownian motions as~$B= \sum_{i=1}^{N} \rho_i W^i$ with~$\sum_{i=1}^N \rho_i^2=1,\,\rho_i\in[-1,1]$ for all~$i\in\llbracket1,N\rrbracket$. The volatility is a function of $(N-1)$ Brownian motions, such that the stock price features one additional and independent source of randomness.
To fit this model into~\eqref{eq:MRT2} we set~$A=S$ and identify~$\phi^i$ with~$\rho_i \sqrt v$. 
We modify slightly the notations to differentiate from the VIX framework: 
the implied volatility is denoted~$\widehat \Ii_{T}$ and the skew~$\widehat \Ss_T$.
We do not consider the curvature in this setting, by lack of an explicit formula. 
The proof of this Proposition and the following Corollary are postponed to Appendix~\ref{app:SPXlimit}.
\begin{proposition}\label{prop:SPXlimit}
	Assume that there exists $H\in(0,\half)$ and a random variable~$X$ such that, for all~$0\le s\le y$, $j\in\llbracket 1,N\rrbracket$, and~$p\ge1$, $X\in L^p$,
	\begin{enumerate}[(i)]
		\item $v_s\le X$ ;
		\item $\Df_s^j v_y\le X (y-s)^{\Hm}$;
		\item $\sup_{s\le T} \EE[u_s^{-p}]  <\infty$;
		\item $\limsup_{T\dto0} \EE\big[(\sqrt{v_T/v_0}-1)^2 \big]=0$.
	\end{enumerate}
	Then the short-time limits of the implied volatility and skew are
$$
\lim_{T\downarrow0} \widehat \Ii_{T} = \sqrt{v_0}
\qquad\text{and}\qquad
\lim_{T\downarrow0} \frac{\widehat\Ss_T}{T^{\Hm}} = \frac{1}{2 v_0}\sum_{j=1}^{N} \rho_j \lim_{T\downarrow0} \frac{\int_0^T \int_s^T \EE\left[ \Df^j_{s} v_y\right]\dy\ds}{T^{H+3/2}}.
$$
\end{proposition}
\begin{remark}\ 
\begin{itemize}
\item The second limit is finite because of Condition~(ii).
\item The one-dimensional version ($N=2$) agrees with~\cite[Theorem 6.3]{ALV07} up to the sign because they derive with respect to the spot~$x$ and not to the log-strike~$k$. 
\end{itemize}
\end{remark}
In the two-factor rough Bergomi model~\eqref{eq:expomodel} we can compute the short-time skew more explicitly. 
Recall from Example~\ref{ex:multirB} that, for all~$t\ge0$, it means setting $N=3$ and defining
\begin{equation*}
\left\{
\begin{array}{rl}
S_t &= \displaystyle S_0 + \int_0^t S_r \sqrt{v_r} \,\D B_r,\\
v_t &= v_0 \left[ \chi \Ee\left(\nu W^{1,H}\right)_{t} + \chib \Ee\left(\eta\left(\rho W^{1,H} + \rrho W^{2,H}\right)\right)_{t} \right],
\end{array}
\right.
\end{equation*}
where $W^{i,H}_t = \int_0^t (t-s)^{\Hm} \,\D W^i_s$, for $i=1,2$ and $B = \sum_{i=1}^3 \rho_i W^i$, with~$W^1,W^2,W^3$ being independent Brownian motions. Hence $W^3$ only influences the asset price but not the variance.
\begin{corollary}\label{coro:SPXlimit}
	In the two-factor rough Bergomi model we have the short-time skew limit
	\begin{equation}\label{eq:SPXlimit}
	\lim_{T\downarrow0} \frac{\widehat\Ss_T}{T^{\Hm}}
	= \frac{\rho_1 \chi\nu + \eta \chib (\rho_1 \rho+\rho_2\rrho)}{2\Hp(1+\Hp)}.
	\end{equation}
\end{corollary}




\subsection{Tips for joint calibration in the two-factor rough Bergomi model}
Assuming we can observe the short-time limits of the spot ATM implied volatility, it grants us~$v_0$ for free, while the slope of its skew gives us~$H$ by~\eqref{eq:SPXlimit}.
Next, we simplify the expressions from Proposition~\ref{prop:expomodel} in the case~$\chi=\half$.  
Call $\Ii_0$, $\Ss_0$ and $\Cc_0$ the three limits of Proposition \ref{prop:expomodel} and denote $H_{\pm} := H\pm\half$, $\alpha:=\eta\rho$, $\beta:=\eta\rrho$.
Introduce further the normalised parameters
$$
\at := \frac{\alpha}{\nu},
\qquad 
\bt := \frac{\beta}{\nu},
$$
so that, denoting $\widetilde\psi(\at, \bt):=\sqrt{(1+\at)^2+\bt^2}$, we have, after simplifications,
\begin{equation*}
\begin{array}{rll}
\Ii_0  & = \frac{\nu\Delta^{\Hm}}{4\Hp} \sqrt{(1+\at)^2+\bt^2}
 & =: \displaystyle \nu C_{I}\widetilde\psi(\at, \bt),\\
\Ss_0 
& = \frac{\nu\Hp\Delta^{\Hm}}{2} \frac{(1+\at)^2 \left[ \frac{1+\at^2}{2H}-\left(\frac{1+\at}{\Hp}\right)^2\right]
+ 2 (1+\at)\bt^2 \left[ \frac{\at}{2H} - \frac{1+\at}{\Hp}\right] +   \at^4 \left[\frac{1}{2H}-\frac{1}{\Hp^2}\right]}{\big( (1+\at)^2+\bt^2 \big)^{3/2}}
 &  =: \displaystyle \nu C_{\Ss}\frac{\Phi_{\Ss}(\at, \bt)}{\widetilde\psi(\at, \bt)^3},\\
\Cc_0
& = \frac{128\nu \Hp^2}{3\Delta^{2H}}
\frac{\Big\{ (1+\at)^3 (1+\at^3) + 3 (1+\at)^2 \at^2 \bt^2 + 3  (1+\at)  \at \bt^4 +  \bt^6 \Big\}}{\big( (1+\at)^2+\bt^2 \big)^{5/2}(1-6H)}
 &  =: \displaystyle \nu C_{\Cc}\frac{\Phi_{\Cc}(\at, \bt)}{\widetilde\psi(\at, \bt)^5},
\end{array}
\end{equation*}
where the constants $C_I, C_\Ss, C_\Cc$ only depend on~$\Delta$ and~$H$.
Provided we can observe an approximation of these three limits, 
we can numerically solve for $\nu,\at,\bt$ in a system with three equations.
Alternatively, since the three quantities have a factor~$\nu$, any quotient of two of them is a function of only~$\at,\bt$, which we can plot and match to observed data.
Both methods allow us to deduce~$\nu,\at,\bt$ in turn yielding~$\eta$ and~$\rho$.
Finally, we are left with~$\rho_1$ and~$\rho_2$ to play with such that the right-hand-side of~\eqref{eq:SPXlimit} matches the market observations.



\section{Proofs}
\subsection{Useful results}
We start by adapting to the multivariate case a well-known decomposition formula and then prove a lemma which will be used extensively in the rest of the proofs. Both proofs build on the multidimensional anticipative It\^o formula~\cite[Theorem 3.2.4]{Nualart06}.

\begin{proposition}[Price decomposition]\label{prop:decomposition}
	Under \HOT, the following decomposition formula holds, for all~$t\in\TT$, for the price~\eqref{eq:PriceVt}, with~$u_t$ defined in~\eqref{eq:defprocesses} and $G:=(\pdx^2 - \pdx)\BS$:
$$
\Pi_t(k) = \EE_{t}\left[\BS(t,\Mf_t,k,u_{t})\right]
+ \half\EE_{t}\left[\int_{t}^{T}\pdx G (s,\Mf_s,k,u_{s})|\Thb_{s}|\ds\right].
$$
\end{proposition}
\begin{proof}
	Define~$\BSh(t,x,k,\sigma^2 (T-t)) := \BS(t,x,k,\sigma)$ and write for simplicity~$\BSh_t:=\BSh\left(t,\Mf_t,k, Y_t \right)=\BS\left(t,\Mf_t,k, u_t \right)$, where we recall that~$Y_t=u_t^2(T-t)$. Note that $\Pi_T= \BSh_T$,  hence
	$\Pi_t=\EE_t\left[\BSh_T\right]$ by no-arbitrage arguments. Thanks to {\HOne} and \HTwo, we then apply a multidimensional anticipative It\^o's formula~\cite[Theorem 3.2.4]{Nualart06} with respect to~$(t,\Mf,Y)$:
	\begin{align*}
	\BS(T, \Mf_T,k, u_{T})
	= \BSh_{T}
	&= \BSh_{t} + \int_t^T \pds \BSh_{s}\ds + \int_t^T \pdx \BSh_{s}
	\left(\D(\phib\bullet\Wb)_{s}-\half \|\phib_s\|^2\ds\right) \\
	&- \int_t^T \pdy \BSh_{s}\|\phib_s\|^2\ds
	+ \half \int_t^T \pdx^2 \BSh_{s}\|\phib_s\|^2\ds + \int_t^T \pdxy \BSh_{s}|\Thb_s|\ds.
	\end{align*}
	with~$\Thb$ in~\eqref{eq:defprocesses}.
	Derivatives of the Black-Scholes price (omitting the argument for simplicity) read
	$$
	\pds \BSh_s= \pds \BS_{s} + \frac{u_s \pdu\BS}{2(T-s)},
	\qquad
	\pdy \BSh_{s} = \frac{\pdu\BS}{2 u_{s} (T-s)},
	\qquad\text{and}\qquad
	G=\frac{\pdu\BS}{u_{s} (T-s)}.
	$$
	Putting everything together, using the Gamma-Vega-Delta relation
	\begin{equation}\label{eq:GVDrelation}
	\frac{\pdsi\BS(t,x,k,\sigma)}{\sigma(T-t)} = ( \pdx^2 - \pdx) \BS(t,x,k,\sigma),
	\end{equation}
	and applying conditional expectation, we obtain
	\begin{equation}\label{eq:lastTerm}
	\Pi_t
	= \EE_t\left[\BS(t,\Mf_t, k,u_{t})\right]
	+\EE_t \left[\int_{t}^{T}\Ll_{\BS} (s,u_{s})\ds\right]
	+\EE_{t}\left[ \int_{t}^{T}\frac{\pd_{xu}\BS(s,\Mf_s,k,u_{s})}{2u_{s}(T-s)}|\Thb_{s}|\ds\right],
	\end{equation}
	where $\Ll_{\BS}(s,u_{s}) := \half\left[u_{s}^2\left(\pdx^2-\pdx\right)+\pds\right]\BS(s,\Mf_s,k,u_{s})$ 
	is the Black-Scholes operator applied to the Black-Scholes function. 
	Since $\Ll_{\BS}(s,u_{s})=0$ by construction and
	$$
	\pd_{x} G(s,x,k,\sigma)=\frac{\E^x \Nn'(d_+(x,k,\sigma))}{\sigma \sqrt{T-s}}\left(1-\frac{d_+(x,k,\sigma)}{\sigma\sqrt{T-s}}\right),
	$$
	the last term in~\eqref{eq:lastTerm} is well defined by \HThree and the proposition follows.
\end{proof}

\begin{lemma}\label{lemma:ItoProd}
	For all~$t\in\TT$, let $J_t := \int_t^T a_s \ds$, for some adapted process~$a\in \mathbb{L}^{1,2}$, and $\mathfrak{L}:=\sum_{i=1}^n c_i \pd_x^i$ be a linear combination of partial derivatives, with weights~$c_i\in\RR$. Then, writing for clarity~$\BS_t:=\BS(t,\Mf_t,\Mf_0,u_{t})$, we have
	\begin{align} \label{eq:ItoProd}
	\EE \left[ \int_0^T \mathfrak{L} \BS_s a_s \ds \right] 
	= \EE\left[ \mathfrak L \BS_0 J_0 
	+ \int_0^T \left(\partial_x^3-\partial_x^2\right) \mathfrak L \BS_s \lvert \Thb_s\lvert J_s \ds 
	+ \int_0^T \partial_x \mathfrak L \BS_s \sum_{k=1}^N \left( \phi_s^k \Df^k_s J_s\right) \ds \right].
	\end{align}	
\end{lemma}
\begin{remark}
We will use this lemma freely below with the justification that the condition $a\in \mathbb{L}^{1,2}$ is always  satisfied thanks to \HOne.
\end{remark}
\begin{proof}
As in the proof of Proposition~\ref{prop:decomposition}, we define~$\BSh(t,x,k,\sigma^2 (T-t)) := \BS(t,x,k,\sigma)$ and write for simplicity~$\BSh_t:=\BSh\left(t,\Mf_t,\Mf_0, Y_t \right)=\BS\left(t,\Mf_t,\Mf_0, u_t \right)$.
Define~$\Ph(t,x,k,y,j):=\mathfrak{L} \widehat \BS(t,x,k,u)j$ and denote~$\Ph\left(t,\Mf_s,\Mf_0,u_s, J_s\right)$ by~$\Ph_t$ for simplicity.
	We then apply the multidimensional anticipative It\^o's formula~\cite[Theorem 3.2.4]{Nualart06} with respect to~$(t,\Mf,Y,J)$:
	\begin{align*}
	\Ph_{T}
	=& \,\Ph_{0} + \int_0^T \pds \Ph_{s}\ds + \int_0^T \pdx \Ph_{s}
	\left(\D(\phib\bullet\Wb)_{s}-\half \norm{\phib_s}^2\ds\right) 
	- \int_t^T \pdy \Ph_{s}\norm{\phib_s}^2\ds\\
	& + \half \int_t^T \pdx^2 \Ph_{s}\norm{\phib_s}^2\ds + \int_t^T \pdxy \Ph_{s}|\Thb_s|\ds + \int_0^T \pd_j \Ph_s \,\D J_s + \int_0^T \pd_{xj} \Ph_s \sum_{k=1}^N \left(\phi^k_s \Df^k_s J_s\right) \ds.
	\end{align*}
	One first notices that~$\Ph_0 = \mathfrak{L} \widehat\BS_0 J_0$ and~$\Ph_T=0$.
	Moreover we observe that~$\int_0^T \pd_j \Ph_s \,\D J_s = -\int_0^T \mathfrak{L} \widehat\BS_s a_s \ds$, which corresponds to the left-hand-side of~\eqref{eq:ItoProd}, and
	\[
	\int_0^T \pd_{xj} \Ph_s \sum_{k=1}^N \left(\phi^k_s \Df^k_s J_s\right) \ds 
	= \int_0^T \pdx \mathfrak{L}\widehat \BS_s \sum_{k=1}^N \left( \phi^k_s \, \Df^k_s J_s \right) \ds.
	\]
	Since $\mathfrak{L}$ is a linear operator the partial derivatives in $s$, $x$ and $u$ cancel as in the proof of Proposition~\ref{prop:decomposition}. That means we are left with
	\begin{align*}
	\int_0^T \mathfrak{L} \widehat\BS_s a_s \ds 
	& = \mathfrak L \widehat\BS_0 J_0 + \int_0^T \left(\partial_x^3-\partial_x^2\right) \mathfrak{L} \widehat\BS_s \lvert \Thb_s\lvert J_s \ds +\int_0^T \pdx \Ph_{s}
	\D(\phib\bullet\Wb)_{s}\\
	& \quad + \int_0^T \pdx \mathfrak{L}\widehat \BS_s \sum_{k=1}^N \left( \phi^k_s \, \Df^k_s J_s \right) \ds.
	\end{align*}
	Since~$\partial_x^n \BS(s,x,u)= \partial_x^n \widehat \BS(s,x,u^2(T-s))$ for any $n\in\mathbb N$, 
summing everything and taking expectations imply the claim.	
\end{proof}

We adapt and clarify~\cite[Lemma 4.1]{ALV07}, yielding a convenient bound for the partial derivatives of~$G$.
For notational simplicity, since $\sigma$ and $T-t$ are fixed, we write
$\varsigma:=\sigma\sqrt{T-t}$ and $\Gf(x,k,\varsigma):=G(t,x,k,\sigma)$.
\begin{proposition}\label{prop:DerivativeBound}
For any $n\in\mathbb{N}$ and $p\in\mathbb{R}$,  there exists $C_{n,p}>0$ independent of~$x$ and~$\varsigma$ such that, for all $\varsigma>0$ and $x\in\mathbb{R}\setminus\left\{0, \frac{\varsigma^2}{2}\right\}$,
	\begin{equation}\label{eq:DiffnG}
	\partial_{x}^{n}\Gf(x,k,\varsigma) \leq 
	\frac{C_{n,p}\,\E^k}{\varsigma^{p+1}}.
	\end{equation}
	
	If $x=0$, then for any $n\in\mathbb{N}$
	the bound~\eqref{eq:DiffnG} holds with $p=n$.
	
	If $x=\half \varsigma^2$, 
	there exists a strictly positive constant~$C_{n}$ independent of~$\varsigma$ such that
	\begin{equation*}
	\partial_{x}^{n}\Gf\left(\frac{\varsigma^2}{2},k,\varsigma\right) =
	\left\{
	\begin{array}{ll}
	\displaystyle \frac{C_n\,\E^k}{\varsigma^{n+1}}, & \text{if } n\text{ is even},\\
	0, & \text{if } n\text{ is odd}.
	\end{array}
	\right.
	\end{equation*}
\end{proposition}

The following simplification (and extension) will be useful later:
\begin{corollary}\label{coro:DerivativeBound}
	For any $n\in\mathbb{N}$, there exists a non-negative~$C_{n,k}$ independent of~$x$ and~$\varsigma$ such that, for all $\varsigma>0$ and $x\in\mathbb{R}$,
	\begin{align*}
	\left|\partial_{x}^{n}\Gf(x,k,\varsigma)\right| \leq 
	\frac{C_{n,k}}{\varsigma^{n+1}}, \qquad
	\left|\partial_k \partial_{x}^{n}\Gf(x,k,\varsigma)\right| \leq 
	\frac{C_{n,k}}{\varsigma^{n+2}}\qquad\text{and}\qquad
	\left|\partial_k^2 \partial_{x}^{n}\Gf(x,k,\varsigma)\right| \leq 
	\frac{C_{n,k}}{\varsigma^{n+3}}.
	\end{align*}
\end{corollary}

\begin{proof}[Proof of Proposition~\ref{prop:DerivativeBound}]
	We first consider the case $k=0$. Since 
	$$
	\Gf(x,0,\varsigma):= (\partial_{xx} - \partial_{x}) \mathrm{BS}(t, x, 0,\varsigma)
	= \frac{1}{\varsigma\sqrt{2\pi}}\exp\left\{x -\frac{1}{2}d_+(x,\varsigma)^2\right\},
	$$
	where
	$d_+(x,\varsigma) :=d_+(x,0,\sigma)= \frac{x}{\varsigma}+\frac{\varsigma}{2}$,
	direct computation (proof by recursion) yields, for any $n\in\mathbb{N}$,
	\begin{equation}\label{eq:DiffGn2}
	\partial_{x}^{n}\Gf(x,0,\varsigma)
	= \exp\left\{-\frac{(\varsigma^2-2x)^2}{8\varsigma^2}\right\}\sum_{j=0}^{n}\alpha_j\frac{P_j(x)}{\varsigma^{2j+1}},
	\end{equation}
	where, for each~$j$, $P_{j}$ is a polynomial of degree~$j$ independent of~$\varsigma$.
	
Since
	$d_+(\frac{\varsigma^2}{2},\varsigma)=\partial_x d_+(\frac{\varsigma^2}{2},\varsigma)=0$,
	$\partial_x^2 d_+(\frac{\varsigma^2}{2},\varsigma) = -\frac{1}{\varsigma^{2}}$,
	the induction simplifies to
	\begin{equation*}
	\left.\partial_{x}^{n}\Gf(x,0,\varsigma)\right|_{x=\frac{\varsigma^2}{2}} =
	\left\{
	\begin{array}{ll}
	\displaystyle \frac{C_n}{\varsigma^{n+1}}, & \text{if } n\text{ is even},\\
	0, & \text{if } n\text{ is odd},
	\end{array}
	\right.
	\end{equation*}
	for some constant $C_n>0$ independent of~$\varsigma$, proving the third statement in the proposition.
	
	Similarly, if $x=0$, simplifications occur which yield for any $n\in\mathbb{N}$,
	$$
	\left.\partial_{x}^{n}\Gf(x,0,\varsigma)\right|_{x=0}
	= \exp\left\{-\frac{\varsigma^2}{8}\right\}
	\sum_{j=0}^{n}\frac{\alpha_j}{\varsigma^{j+1}}
	= \frac{1}{\varsigma^{n+1}}\exp\left\{-\frac{\varsigma^2}{8}\right\}
	\sum_{j=0}^{n}\alpha_j \varsigma^{n-j},
	$$
	and the second statement in the proposition follows.
	
	Finally, in the general case $x\in\mathbb{R}\setminus\left\{0,\frac{\varsigma^2}{2}\right\}$, 
	we can rewrite~\eqref{eq:DiffGn2} for any $p\in\mathbb{R}$ as
	$$
	\partial_{x}^{n}\Gf(x,0,\varsigma)
	= \frac{1}{\varsigma^{p+1}}\exp\left\{-\frac{(\varsigma^2-2x)^2}{8\varsigma^2}\right\}
	\sum_{j=0}^{n}\alpha_k P_j(x) \varsigma^{p-2j}
	 =: \frac{1}{\varsigma^{p+1}}\exp\left\{-\frac{(\varsigma^2-2x)^2}{8\varsigma^2}\right\}
	H_{n,p}(x,\varsigma).
	$$
	For each~$n \in \mathbb{N}, p\in\mathbb{N}$, $H_{n,p}$ is a 
	two-dimensional function only consisting of 
	powers of $\varsigma^2$ and $x^2/\varsigma^2$.
	Since the exponential factor contains these very same terms, there exists then a strictly positive constant~$C_{n,p}$,
	independent of~$x$ and~$\varsigma$, such that 
	$$ \exp\left\{-\frac{(\varsigma^2-2x)^2}{8\varsigma^2}\right\} H_{n,p}(x,\varsigma) \leq C_{n,p},$$
	proving the proposition in the case~$k=0$.
	
	The case $k\in\RR$ follows directly by noticing that~$\Gf(x,0,\varsigma) = \Gf(x-k,0,\varsigma) \E^k$. Finally, since~$\partial_k d_+(x,k,\sigma) = - \partial_x d_+(x,k,\sigma)$ and~$\partial_k^2 d_+(x,k,\sigma) = - \partial_x^2 d_+(x,k,\sigma)$, the same simplifications occur when taking a partial derivative with respect to~$k$ instead of~$x$.
\end{proof}


\subsection{Proofs of the main results}\label{app:proofsmain}

\subsubsection{Proof of Theorem~\ref{thm:level}: Level}\label{sec:prooflevel}

To prove this result, we draw insights from the proofs of~\cite[Theorem 8]{AGM18} and~\cite[Proposition 3.1]{AS19}. 
By definition 
$$
\Ii_{T}= \BS^{\lef}(0,\Mf_0,\Mf_0,\Pi_0) =: \BSB(\Pi_0),
$$
and we write $\BST(x) := \BS(0,x,x,u_0)$.
Using Proposition~\ref{prop:decomposition} at time~$0$ we see that $\Pi_0=\Gamma_T$ where
$$
\Gamma_t := \EE\left[\BST(\Mf_0) + \half \int_0^t \pdx G(s,\Mf_s,\Mf_0,u_s)|\Thb_{s}|\ds\right],
\qquad\text{for }t\in\TT,
$$
which is a deterministic path.
The fundamental theorem of integration reads
\begin{align*}
\Ii_{T}=\BSB(\Gamma_T) 
= \BSB(\Gamma_0) + \int_0^T \pdt \BSB(\Gamma_t) \dt 
&= \BSB(\Gamma_0) + \int_0^T  \BSB'(\Gamma_t) \pdt \Gamma_t \dt\nonumber \\
& = \BSB(\Gamma_0) + \half \int_0^T \BSB'(\Gamma_t) \, \EE\big[|\Thb_{t}|\pdx G_t \big] \dt,
\label{eq:intBSB}
\end{align*}
where $G_t:= G(t,\Mf_t,\Mf_0,u_t)$.
We can deal with the integral by computing $\BSB'$ and $\partial_x G$ explicitly 
\begin{align*}
& \BSB'(\Gamma_t) = \left(\E^{\Mf_0} \Nn'\left(d_+\left(\Mf_t,\Mf_0,\BSB(\Gamma_t)\right)\right) \sqrt{T-t} \right)^{-1}, \\
& \partial_x G (s,x,k,\sigma) = \frac{\E^{x} \Nn'\big(d_+(x,k,\sigma)\big)}{\sigma\sqrt{T-s}} \left( 1 -\frac{d_+(x,k,\sigma)}{\sigma \sqrt{T-s}}\right).
\end{align*}
Since~$\Gamma:\RR_+\to\RR$ and~$\BSB:\RR\to\RR$ are continuous, the following is uniformly bounded for all~$T\le1$
\[
\frac{\Nn' \big(d_+(\Mf_s,\Mf_0,u_0)\big)}{\Nn'\left(d_+\left(\Mf_s,\Mf_0,\BSB(\Gamma_s) \right)\right)}
 = \exp\left\{\frac{1}{8}\left( (T-s) \,\BSB(\Gamma_s)^2 - \ut_0^2\right) \right\}.
\] 
Therefore, we obtain by {\HFour}
\[
\lim_{T\downarrow0} \EE\left[\int_0^T \BSB'(\Gamma_t) |\Thb_{t}|\pdx G_t \dt \right]
 = \lim_{T\downarrow0} \EE\left[  \int_0^T \frac{\Nn'\left(d_+\left(\Mf_t,\Mf_0,u_t\right) \right)}{\Nn' \left(d_+\left(\Mf_t,\Mf_0,\BSB(\Gamma_t)\right) \right)} \frac{\lvert \Thb_t\lvert}{2 \ut_t^2 \sqrt{T}} \dt
\right] =0.
\]

Since $\Gamma_0=\EE\left[\BST(\Mf_0)\right]$ and~$u_0=\BSB\left(\BST(\Mf_0)\right)$, we have
\begin{equation}\label{eq:diffBSBu}
\BSB(\Gamma_0)
= \BSB\left(\EE\left[\BST(\Mf_0)\right]\right) - \EE[u_0 - u_0]
= \EE\left[ \BSB\Big(\EE\left[\BST(\Mf_0)\right] \Big) - \BSB\Big(\BST(\Mf_0) \Big)\right] + \EE[u_0].
\end{equation}
The Clark-Ocone formula yields
$\BST(\Mf_0) =\EE\left[\BST(\Mf_0)\right]+\sum_{i=1}^N\int_0^T \EE_s\left[\Df^{i}_s \BST(\Mf_0)\right] \D W^i_s$ and, by the Gamma-Vega-Delta relation~\eqref{eq:GVDrelation} we have
\begin{equation}\label{eq:pdsigBS}
\pd_\sigma \BS (0,x,x,\sigma) =\exp\left\{x-\frac{\sigma^2}{8}T\right\}\sqrt{\frac{T}{2\pi}},
\end{equation}
which in turn implies
\begin{equation}
U^i_s := \EE_s\left[\Df^{i}_s \BST(\Mf_0)\right]
= \EE_s \left[ \frac{\pdsi \BST(\Mf_0)}{2 u_0 \sqrt T} \int_0^T \Df^i_s \norm{\phib_r}^2 \dr \right] 
 = \EE_s \left[ \frac{\E^{\Mf_0 - \frac{1}{8}\ut_0^2}}{2\sqrt{2\pi} u_0} \int_0^T  \Df^i_s \norm{\phib_r}^2 \dr \right].
\label{eq:Ui}
\end{equation}
Define~$\Lambda_r:=\EE_r \Big[\BST(\Mf_0)\Big]$ 
so that the difference we are interested in from~\eqref{eq:diffBSBu} reads, after applying the standard It\^o's formula,
\begin{equation}\label{eq:LDiff}
\BSB(\Gamma_0)-\EE[u_0] = \EE\Big[ \BSB ( \Lambda_0 ) - \BSB (\Lambda_T ) \Big]
 = -\sum_{i=1}^N \EE\left[\int_0^T \BSB'(\Lambda_s)  U^i_s \D W^i_s+\half\int_0^T\BSB''(\Lambda_s)  (U^i_s)^2 \ds\right].
\end{equation}
The stochastic integral above has zero expectation by the same argument as~\cite[Proposition 3.1]{AS19}.
Moreover, {\HFive} states that $\ut_0$ is dominated almost surely by $Z\in L^p$, and therefore so are~$\Lambda$ and
\[
\BSB''(\Lambda_s) = \frac{\BSB(\Lambda_s)}{4 \E^{2\Mf_s} \Nn'\Big(d_+\big(\Mf_s,\Mf_0,\BSB(\Lambda_s)\big)\Big)^2},
\]
by continuity.
Plugging in the expression of~$U^i$ from~\eqref{eq:Ui}, we apply {\HFive} to conclude that the second integral of~\eqref{eq:LDiff} tends to zero. 

\subsubsection{Proof of Theorem~\ref{thm:skew}: Skew}\label{sec:proofskew}

This proof follows similar arguments as~\cite[Proposition 5.1]{ALV07}. We recall that~$\Pi_0(k)= \BS\big(0,\Mf_0,k,\Ii_{T}(k)\big)$. On the one hand, by the chain rule we have 
\begin{equation}\label{eq:SkewChainRule}
\pdk \Pi_0(k) = \pdk\BS\big(0,\Mf_0,k,\Ii_{T}(k)\big)+\pdsi\BS\big(0,\Mf_0,k,\Ii_{T}(k)\big) \pdk \Ii_{T}(k).
\end{equation}
On the other hand, the decomposition obtained in Proposition~\ref{prop:decomposition} yields 
\begin{equation}\label{eq:SkewDecomp}
\pdk \Pi_0(k)
= \EE\big[\pdk\BS(0,\Mf_T,k,u_{0})\big]
+\EE\left[ \int_{0}^{T} \half \pdxk G
(s,\Mf_s,\Mf_0,u_s)|\Thb_{s}|\ds\right].
\end{equation}
Equating~\eqref{eq:SkewChainRule} and~\eqref{eq:SkewDecomp} gives
\begin{equation}\label{eq:SkewExpression}
\pdk \Ii_{T}(k)
=\frac{\EE\left[\pdk \BS(0,\Mf_T,k,u_{0})\right] - \pdk \BS(0,\Mf_0,k,\Ii_{T}(k))}{\pdsi \BS(0,\Mf_0,k,\Ii_{T}(k))}
+\frac{\EE\left[ \int_{0}^{T}\pdxk G(s,\Mf_s,\Mf_0,u_s)|\Thb_{s}|\ds\right]}{2\pdsi\BS(0,\Mf_0,k,\Ii_{T}(k))},
\end{equation}
which in particular also holds for $k=\Mf_0$. Performing simple algebraic manipulations and using the derivatives of the Black-Scholes function ATM as in~\cite[Proposition 5.1]{ALV07} we obtain that the difference (remember we drop the $k$-dependence in~$\Ii_{T}$ when ATM)
$$
\EE\big[\pdk\BS(0,\Mf_0,\Mf_0,u_{0}) \big]
- \pdk\BS(0,\Mf_0,\Mf_0,\Ii_{T})
= \half \EE\left[\int_0^T \half\pdx G(s,\Mf_s,\Mf_0,u_s)|\Thb_{s}|\ds\right],
$$
which in turn using~\eqref{eq:SkewExpression} yields
\begin{equation}\label{eq:SkewExpressionII}
\pdk \Ii_{T}
= \frac{\EE\left[ \int_{0}^{T} L(s,\Mf_s,\Mf_0,u_s)|\Thb_{s}|\ds\right]}
{\pdsi\BS\Big(0,\Mf_0,\Mf_0,\Ii_{T} \Big) },
\end{equation}
where $L:=(\half+\pdk)\half\pdx G$. 
We denote $L_s:=L(s,\Mf_s,\Mf_0,u_s)$ for simplicity and apply Lemma~\ref{lemma:ItoProd} to~$L_s\, \int_s^T |\Thb_r|\dr$, which yields
\begin{align*}
\EE \left[ \int_0^T L_s |\Thb_s|\ds  \right] 
& = \EE \left[L_0 \int_0^T|\Thb_s| \ds \right]
+  \EE \left[ \int_0^T (\pdx^3 - \pdx^2) L_s |\Thb_s| \left( \int_s^T |\Thb_r|\dr \right) \ds \right] \\
& + \EE \left[ \int_0^T \pdx L_s \sum_{j=1}^N \left\{ \phi^j_s \, \Df^j_s \left( \int_s^T |\Thb_r|\dr \right) \right\} \ds \right]
=: R_1 + R_2 + R_3.
\end{align*}
We combine~\eqref{eq:pdsigBS}
with the bound~$\partial_k \partial_x^n G(t,x,k,\sigma) \le C \big(\sigma\sqrt{T-t}\big)^{-n-2}$ from Corollary~\ref{coro:DerivativeBound} to obtain
\begin{align*}
\frac{R_2}{\pd_\sigma \BS(0,\Mf_0,\Mf_0,\Ii_{T})}
&\le \frac{C}{\sqrt{T}}\EE\left[\int_0^T  \frac{|\Thb_s|}{\ut_s^{6}} \left( \int_s^T |\Thb_r|\dr \right) \ds\right], \\
\frac{R_3}{\pd_\sigma \BS(0,\Mf_0,\Mf_0,\Ii_{T})}
&\le\frac{C}{\sqrt{T}}  \EE\left[\int_0^T \frac{1}{\ut_s^{4}} \sum_{j=1}^N \left\{ \phi^j_s \Df^j_s \left( \int_s^T |\Thb_r|\dr \right) \right\} \ds\right],
\end{align*}
and both converge to zero by \HSix.
We are left with~$R_1$. From Appendix~\ref{app:partiald}, we have
$$
L(0,x,x,u) = \frac{\exp\{x-\frac{u^2}{8}T\}}{u\sqrt{2\pi T}} 
\left( \frac14 + \frac{1}{2u^2 T} \right),
$$
and therefore by~\eqref{eq:pdsigBS},
$$
\frac{L_0}{\pd_\sigma \BS(0,\Mf_0,\Mf_0,\Ii_{T})}
= \left( \frac14 + \frac{1}{2u_0^2 T} \right) \frac{1}{u_0 T} \exp\left\{ -\frac{T}{8} \Big(u_0^2 - \Ii_{T}^2 \Big)\right\}.
$$
This yields
\[
\displaystyle  \frac{R_1}{T^{\lambda}\pd_\sigma \BS(0,\Mf_0,\Mf_0,\Ii_{T})}
= \EE\left[ \left(\frac{\ut_0^2}{2}+1\right) \exp\left\{ -\frac{T}{8} \Big(u_0^2 - \Ii_{T}^2 \Big)\right\} \Kf_T\right],
\]
where $\displaystyle\Kf_T:=\frac{\int_0^T |\Thb_s|\ds}{2 T^{\half+\lambda} \ut_0^3}$. Furthermore,
\[
\sup_{\omega\in\Omega} \left|\exp\left\{ -\frac{T}{8} \Big(u_0^2 - \Ii_{T}^2 \Big)\right\} -1\right|
= \sup_{\omega\in\Omega} \left| \exp\left\{ \frac18 \left(T \Ii_{T}^2 -  \ut_0^2 \right)\right\} -1\right|
\]
not only is finite but converges to zero as~$T$ goes to zero.
Hence,
\[
\lim_{T\dto0} \EE\left[ \left(\frac{\ut_0^2}{2}+1\right) \exp\left\{ -\frac{T}{8} \Big(u_0^2 - \Ii_{T}^2 \Big)\right\} \Kf_T\right] = \lim_{T\dto0} \EE\left[ \left(\frac{\ut_0^2}{2}+1\right)  \Kf_T\right].
\]
We can finally conclude by {\HSeven} that 
\[
\displaystyle \lim_{T\downarrow 0} \frac{R_1}{T^{\lambda}\pd_\sigma \BS(0,\Mf_0,\Mf_0,\Ii_{T})}
= \lim_{T\downarrow 0} \EE[\Kf_T],
\]
which has a finite limit.

\subsubsection{Proof of Theorem~\ref{thm:curvature}: Curvature}\label{sec:proofcurvature}
\textbf{Step 1.} Let us start by simply taking a second derivative with respect to~$k$
and write~$\BS(\Mf_0,\Ii_{T}(k))$ as short for~$\BS(0,\Mf_0,\Mf_0,\Ii_{T}(k))$.
\begin{align*}
\pdk  \Big( \pdsi\BS\big(\Mf_0,\Ii_{T}(k)\big) \,\pdk \Ii_{T}(k) \Big) 
 & = \pdsi\BS\big(\Mf_0,\Ii_{T}(k)\big)  \pdk^2 \Ii_{T}(k) \\
&\quad+\Big[ \pd_{k\sigma}\BS\big(\Mf_0, \Ii_{T}(k)\big) 
+ \pdsi^2\BS\big(\Mf_0,\Ii_{T}(k)\big)  \pdk \Ii_{T}(k) \Big] \pdk \Ii_{T}(k).
\end{align*}
Taking the derivative with respect to~$k$ in~\eqref{eq:SkewExpressionII} and equating with the above formula yields
\begingroup
\addtolength{\jot}{.5em}
\begin{align*}
\pdk^2 \Ii_{T}(k) \pdsi\BS(\Mf_0,\Ii_{T})
 & = 
- \pdsi^2\BS\big(\Mf_0, \Ii_{T}(k)\big) \pdk \Ii_{T}(k)^2
 - \pd_{k\sigma}\BS\big(\Mf_0, \Ii_{T}(k)\big) \pdk \Ii_{T}(k)\\
&\quad + \EE \left[ \int_0^T \pdk L(s,\Mf_s,\Mf_0,u_{s} ) |\Thb_s| \ds \right]
=: T_1 + T_2 + T_3.
\end{align*}
\endgroup
A similar expression is presented in~\cite{AL17} and we notice that 
$T_1$ and~$T_2$ in the expression above are, after multiplying them by~$T^{-\lambda}$, identical to those from~\cite[Equation (25)]{AL17} and can therefore be dealt with in the same way. Step 1 shows that~$T^{-\lambda} T_1$ tends to zero as~$T\downarrow 0$  and step 2  yields $T_2= - \half \pdk \Ii_{T}(k)$.

\textbf{Step 2.}
Recall that $L = \half \left(\half + \pdk \right) \pdx G$. 
We  need the anticipating It\^o's formula (Lemma~\ref{lemma:ItoProd}) twice on~$T_3$. 
Indeed, even though the bound on~$\pdx^n G$ worsens as~$n$ increases, it is more than compensated by the additional integrations. 
The terms with the more integrals (i.e. more regularity) tend to zero as~$T$ goes to zero by {\HEight} and we compute the others in closed-form. 
For clarity we write~$L_s = L (s,\Mf_s,\Mf_0,u_s )$ for all~$s\ge0$.
By a first application of Lemma~\ref{lemma:ItoProd} on~$\pdk L_s \int_s^T |\Thb_s| \ds$  we obtain
\begin{align*}
T_3 
= & \,  \EE \left[\pdk L_0 \int_0^T |\Thb_s| \ds\right]
+  \EE \left[ \int_0^T (\pdx^3 - \pdx^2) \pdk  L_s |\Thb_s| \left( \int_s^T |\Thb_r|\dr \right) \ds \right] \\
& + \EE \left[ \int_0^T \pdxk  L_s \sum_{j=1}^N \left\{ \phi^j_s \, \Df^j_s \left( \int_s^T |\Thb_r| \dr \right) \right\} \ds \right]
=: S_1 + S_2 + S_3. 
\end{align*}
To deal with~$S_2$, we apply Lemma~\ref{lemma:ItoProd} again on
$
(\pdx^3 - \pdx^2) \pdk  L_s \int_s^T |\Thb_r|\left( \int_r^T |\Thb_{y}|\dy \right) \dr =: H_s Z_s$,
yielding
\[
S_2 
= \EE \left[ H_0 Z_0 
+ \int_0^T (\pdx^3 - \pdx^2) H_s |\Thb_s| Z_s \ds
+ \int_0^T \pdx H_s \sum_{j=1}^N \left( \phi^j_s \Df^j_s Z_s \right) \ds
\right]
=: S_2^a + S_2^b + S_2^c.
\]
We will deal with those terms in the last step.
Regarding~$S_3$, we apply Lemma~\ref{lemma:ItoProd} once more to
$$
\pdxk L_s \int_s^T \sum_{j=1}^N \left\{ \phi^j_r \, \Df^j_r \left( \int_r^T |\Thb_{y}| \dy \right) \right\} \dr
=: \widetilde{H}_s \widetilde{Z}_s,
$$
and obtain
$$
S_3 
= \EE \left[ \widetilde{H}_0 \widetilde{Z}_0
+ \int_0^T (\pdx^3 - \pdx^2) \widetilde{H}_s |\Thb_s| \widetilde{Z}_s \ds
+ \int_0^T \pdx \widetilde{H}_s \sum_{j=1}^N \left( \phi^j_s \Df^j_s \widetilde{Z}_s \right) \ds \right]
=: S_3^a + S_3^b + S_3^c.
$$

\textbf{Step 3.} We now evaluate the derivative at~$k=\Mf_0$ and drop the~$k$ dependence. 
To summarise,
$$
\pdk^2 \Ii_{T} = \frac{T_1 + T_2 + S_1 + S_2^a + S_2^b + S_2^c + S_3^a + S_3^b + S_3^c }{\pd_\sigma \BS \big(0,\Mf_0,\Mf_0,\Ii_{T}\big)},
$$
where
\begin{align*}
S_1 & = \EE \left[\pdk L_0 \int_0^T |\Thb_s| \ds \right],\\
S_2^a & = \EE \left[H_0 \int_0^T |\Thb_r| \left( \int_r^T |\Thb_{y}|\dy \right) \dr\right], \\
S_3^a &= \EE \left[\widetilde{H}_0 \int_0^T \sum_{j=1}^N \left\{ \phi^j_r \, \Df^j_r \left( \int_r^T |\Thb_{y}| \dy \right) \right\} \dr\right], \\
S_2^b &=  \EE\left[\int_0^T (\pdx^3 - \pdx^2) H_s |\Thb_s|
\left( \int_s^T \lvert \Thb_r\lvert \left( \int_r^T |\Thb_{y}|\dy \right) \dr \right) \ds\right], \\
S_2^c &= \EE \left[\int_0^T \pdx H_s \sum_{j=1}^N \left( \phi^j_s \Df^j_s \left( \int_s^T |\Thb_r| \left( \int_r^T |\Thb_{y}|\dy \right) \dr \right) \ds \right) \ds\right],\\
S_3^b &= \EE \left[\int_0^T (\pdx^3 - \pdx^2) \widetilde{H}_s |\Thb_s|
\int_s^T \sum_{j=1}^N \left\{ \phi^j_r \, \Df^j_r \left( \int_r^T |\Thb_{y}| \dy \right) \right\} \dr \ds\right], \\
S_3^c &= \EE \left[\int_0^T \pdx \widetilde{H}_s \sum_{k=1}^N \left\{ \phi^k_s \Df^k_s \left( \int_s^T \sum_{j=1}^N \left\{ \phi^j_r \, \Df^j_r \left( \int_r^T |\Thb_{y}| \dy \right) \right\} \dr \right) \right\} \ds\right].
\end{align*}
We recall once again the bound~$\partial_x^n G(t,x,k,\sigma) \le C \big(\sigma\sqrt{T-t}\big)^{-n-1}$ as~$T-t$ goes to zero. 
We observe that~$H$ and~$\widetilde H$ consist of derivatives of~$G$ up to the $6$th and $4$th order respectively, therefore~$S_2^b, S_2^c,S_3^b,S_3^c$ tend to zero by \HEight.
In order to deal with $S_1$, $S_2^a$ and $S_3^a$, we use the explicit partial derivatives from Appendix~\ref{app:partiald} and~\eqref{eq:pdsigBS} and,  as in the proof of Theorem~\ref{thm:skew}, {\HNine} entails that only the higher derivatives of~$\ut_0$ remain in the limit.
\begin{align*}
\lim_{T\downarrow 0} \frac{S_1}{T^{\lambda}\pd_\sigma \BS\big(0,\Mf_0,\Mf_0,\Ii_{T}\big)}
& = \lim_{T\downarrow 0} \frac{1}{T^{\lambda}}\frac{\EE\left[ \pdk L_0 \int_0^T |\Thb_s| \ds \right]}{\pd_\sigma \BS\big(0,\Mf_0,\Mf_0,\Ii_{T}\big)}\\
& = \lim_{T\downarrow 0} \frac{1}{T^{\lambda}}\EE \left[ \frac{1}{u_0  T} \left( \frac18 + \frac{1}{2 u_0^2 T}\right) \int_0^T |\Thb_s| \ds \right] \\
& =  \lim_{T\downarrow 0}  \EE \left[ \frac{1}{2 \ut_0^3 \,T^{\half+\lambda}} \int_0^T |\Thb_s| \ds \right]
= \lim_{T\downarrow 0} \frac{\Ss_T}{T^{\lambda}}.\\
\lim_{T\downarrow 0} \frac{S^a_2}{T^{\lambda}\pd_\sigma \BS\big((0,\Mf_0,\Mf_0,\Ii_{T}\big))} 
& = \lim_{T\downarrow 0} \frac{1}{T^{\lambda}} \EE \left[ \frac{Z_0}{\ut_{0} \sqrt{T}} \left(- \frac{15}{2\ut_0^6} - \frac{3}{2\ut_0^4} - \frac{5}{32\ut_0^2} - \frac{1}{64} \right) \right]\\
& = \lim_{T\downarrow 0} \EE \left[ \frac{-15}{2\ut_0^7 T^{\half+\lambda}} \int_0^T |\Thb_r|\left( \int_r^T |\Thb_{y}|\dy \right) \dr\right].\\
\lim_{T\downarrow 0} \frac{S^a_3 }{T^{\lambda}\pd_\sigma \BS\big(0,\Mf_0,\Mf_0,\Ii_{T}\big)}
& = \lim_{T\downarrow 0} \frac{1}{T^{\lambda}}\EE \left[ \frac{\widetilde{Z}_0}{\ut_0 \sqrt T} \left( \frac{3}{2 \ut_0^4} + \frac{3}{8\ut_0^2} + \frac{1}{16} \right) \right]\\
& = \lim_{T\downarrow 0} \frac32 \EE \left[ \frac{1}{\ut_0^5 T^{\half+\lambda}} \int_0^T \sum_{j=1}^N \left\{ \phi^j_r \, \Df^j_r \left( \int_r^T |\Thb_{y}|\dy \right) \right\} \dr \right].
\end{align*}
Hence to conclude, the claim follows from
\[
\lim_{T\downarrow 0} \frac{\Cc_T}{T^{\lambda}}
= \lim_{T\downarrow 0} \frac{T_2 + S_1 + S_2^a + S_3^a}{T^{\lambda} \pd_\sigma \BS(\Mf_0,\Ii_{T})}.
\]


\subsection{Proof of Proposition~\ref{prop:genmodel}: VIX asymptotics}\label{sec:VIXproof}

In this section, we will repeatedly interchange Malliavin derivative and conditional expectation, 
justified by~\cite[Proposition 1.2.8]{Nualart06}.

\begin{proposition}\label{prop:condsaresufficient}
	In the case where~$A=\VIX$, the conditions in~{\CBar} imply assumptions {\HBarlg} for any~$\lambda\in(-\half,0]$ and~$\gamma\in(-1, 3H-\half]$.
\end{proposition}
\begin{proof}
	We write~$a\lesssim b$ when there exists~$X\in L^p$ such that~$a\le Xb$ almost surely, and~$a\approx b$ if~$a\lesssim b$ and~$b\lesssim a$. 
	{\HOne} is granted by the first point of {\CTwo} and {\HTwo} corresponds to {\COne}.
	Since~$1/M$ is dominated, then so is~$1/\VIX$. We then have, for~$i=1,2$ and by Cauchy-Schwarz,
	\[
	m^i_y = \EE_y\left[\frac{\int_T^{T+\Delta} \Df^i_y v_r \dr}{2\Delta \VIX_T}\right]
	\lesssim \int_T^{T+\Delta} (r-y)^{\Hm} \dr = \frac{(T+\Delta-y)^{\Hp} - (T-y)^{\Hp}}{\Hp}.
	\]
	If~$H<\half$ then the incremental function~$x\mapsto (x+\Delta)^{\Hp} -x^{\Hp}$ is decreasing by concavity. For~$j=1,2$ and~$t\le s$, this implies by domination of $1/M$ that~$\phi^i\approx m^i$ is also dominated and
	\begin{align} \label{eq:Dphi}
	\Df^j_s \phi_y^i
	&= \frac{\Df^j_s m^i_y}{M_y} - \frac{m^i_y \Df^j_s M_y}{M_y^2}
	\lesssim\int_T^{T+\Delta} (r-y)^{\Hm} (r-s)^{\Hm} \dr + \int_T^{T+\Delta} (r-y)^{\Hm} \dr \int_T^{T+\Delta} (r-s)^{\Hm} \dr \nonumber\\
	&\le  \frac{\Delta^{2H}}{2H} + \frac{\Delta^{H+1}}{\Hp^2}.
	\end{align}
	Combining these two estimates we obtain
	\[
	\Theta_s^j =2 \phi_s^j \, \int_s^T \left(\sum_{i=1}^N \phi_y^i \Df_s^j\phi^i_y \right) \dy \lesssim T-s.
	\]
	It is clear by now that indices and sums do not influence the estimates, hence we informally drop them for more clarity and continue with the higher derivatives:
	\[
	\Df_t \Theta_s =\Df_t \phi_s \int_s^T \phi_r \Df_s\phi_r \dr + \phi_s \int_s^T \Df_t \phi_r \Df_s \phi_r \dr + \phi_s \int_s^T \phi_r \Df_t\Df_s \phi_r\dr,
	\]
	where the first and second terms behave like~$T-s$. 
	For~$t\le s\le y\le T$, we deduce from~\eqref{eq:Dphi} that $\Df_t\Df_s \phi_y$ consists of five terms; four behave like~$(T-s)$, and only one features three derivatives:
	\begin{align*}
	\Df_t\Df_s m_y \lesssim \int_T^{T+\Delta} (r-s)^{\Hm} (r-t)^{\Hm} (r-y)^{\Hm}\dr &\le \int_T^{T+\Delta} (r-y)^{3H-\frac{3}{2}} \dy\\
	&\approx (T+\Delta-y)^{3H-\half} -(T-y)^{3H-\half}.
	\end{align*}
	If~$H\ge \frac{1}{6}$ concavity implies $\Df_t \Theta_s \lesssim (T-s) $. Otherwise, if~$H<\frac{1}{6}$,
	\[
	\Df_t \Theta_s \lesssim (T-s) + \Big[(T+\Delta-s)^{3H+\half}-\Delta^{3H+\half}\Big] + (T-s)^{3H+\half}\le (T-s) + 2(T-s)^{3H+\half}.
	\]
	When looking at the second derivative of~$\Theta$,  
	the first and second terms behave as~$(T-s)$ and~$\Df_t\Theta_s\lesssim (T-s)+(T-s)^{(3H+\half)\wedge 1}$ respectively, hence we focus on $\int_s^T \Df_w \Df_t \Df_s \phi_y \dy$, where the new term is
	\[
	\Df_w \Df_t \Df_s m_y \approx \int_T^{T+\Delta} (r-w)^{\Hm} (r-s)^{\Hm} (r-t)^{\Hm} (r-y)^{\Hm}\dr \lesssim (T+\Delta-y)^{4H-1} -(T-y)^{4H-1}.
	\]
	If~$H\ge \frac{1}{4}$, then $\Df_t \Theta_s \lesssim (T-s)$ by concavity. 
	Otherwise, when~$H<\frac{1}{4}$,
	\begin{align*}
	\Df_w \Df_t \Theta_s &\lesssim (T-s) +(T-s)^{(3H+\half)\wedge 1} + \Big[(T+\Delta-s)^{4H}-\Delta^{4H}\Big]+  (T-s)^{4H}\\ &\le (T-s) + (T-s)^{(3H+\half)\wedge 1}+ 2 (T-s)^{4H},
	\end{align*}
	where the last inequality holds by yet again the same concavity argument. 
	
	This yields a rule for checking that the quantities in our assumptions indeed converge.
	We summarise the above estimates in the case~$H\le\half$: there exists~$Z\in L^p$ such that for~$s\le T$ and~$T$ small enough,
	\begin{align*}
	\Thb_s \le Z (T-s), \qquad
	\Df \Thb_s \le Z (T-s)^{(3H+\half)\wedge1}, \qquad
	\Df \Df \Thb_s \le Z (T-s)^{(4H)\wedge1}
	\end{align*} 
hold almost surely.
	Thanks to Cauchy-Schwarz inequality we can disentangle the numerators (integrals and derivatives of~$\Thb$) and denominators (powers of~$u$) of the assumptions, which are both uniformly bounded in~$L^p$. 
	We can easily deduce that 
	{\HThree, \HFour, \HFive, \HSix, \HEight} are satisfied (convergence to zero). In {\HSeven}, $\EE[\Kf_T]$ behaves as~$T^{-\lambda}$, hence it converges for any~$\lambda\in(-\half,0]$, and the uniform~$L^2$ bound is satisfied thanks to {\CThree}. 
Moreover, in the limit  the first term in {\HNine} behaves as~$T^{-\gamma}$ and the second behaves as~$T^{3H-\half-\gamma}$, therefore both assumptions are satisfied for any~$\lambda\in(-\half,0]$ and~$\gamma\in(-1, 3H-\half]$. Similarly, {\CThree} ensures the uniform $L^2$ bounds.
\end{proof}
\subsubsection{Convergence lemmas}
We require some preliminaries before diving into the computations.
We tailor three versions of integral convergence fitted for our purposes 
and essential to compute the limits in Theorems~\ref{thm:level}, \ref{thm:skew}, \ref{thm:curvature}. 
The conditions they require hold thanks to the continuity of~\Ffour.
Recall the local Taylor theorem:
if a function~$g(\cdot)$ is continuous on $[0, \delta]$ for some $\delta>0$, then there exists a continuous function~$\ep(\cdot)$ on $[0, \delta]$ with $\lim_{x\downarrow 0}\ep(x)=0$ such that
$g(x) = g(0) + \ep(x)$ for any $x \in [0, \delta]$.
\begin{lemma}\label{lemma:MVT1}
If $f:\RR_+^2\to\RR$ is such that~$f(T,\cdot)$ is continuous on $[0,\delta_0]$ for some $\delta_0>0$ and $\lim\limits_{T\dto0}f(T,0)=f(0,0)$, then
	\begin{equation}\label{eq:MVT1}
	\lim_{T\downarrow0}\frac1T\int_0^T f(T,y)\,\dy =f(0,0).
	\end{equation}
\end{lemma}
\begin{proof}
	For $T<\delta_0$, we can write
	\begin{align*}
	\frac1T \int_0^T f(T,y) \dy = \frac1T \int_0^T [f(T,0) + \ep_0(y) ] \dy = f(T,0) +  \frac1T \int_0^T\ep_0(y) \dy,
	\end{align*}
	where the function $\ep_0$ is continuous on $[0,\delta_0]$ and converges to zero at the origin. Hence, for any~$\eta_0>0$, there exists $\widetilde \delta_0>0$ such that, for any~$y\le\widetilde \delta_0$, 
	$|\ep_0(y)|<\eta_0$. For all~$T<\widetilde \delta_0 \wedge \delta_0$,
	$$
	\bigg\lvert\frac1T \int_0^T \ep_0(y) \dy \bigg\lvert \le \eta_0.
	$$
	Since $\eta_0$ can be taken as small as desired, the fact that~$\lim_{T\dto0} f(T,0)=f(0,0)$ concludes the proof.
\end{proof}

\begin{lemma}\label{lemma:MVT2}
	Let $f:\RR^3_+\to\RR$ be such that, for each~$y\le T$, $f(T,y,\cdot)$ is continuous on~$[0,\delta_0]$ with $\delta_0>0$, $f(T,\cdot,0)$ is continuous on $[0,\delta_1]$ with $\delta_1>0$ and $\lim_{T\dto0}f(T,0,0)=f(0,0,0)$.
Then
	\begin{equation}\label{eq:MVT2}
	\lim_{T\downarrow0} \frac{1}{T^2} \int_0^T \int_0^y f(T,y,s) \ds \dy = \frac{f(0,0,0)}{2}.
	\end{equation}
\end{lemma}
\begin{proof}
	For $T<\delta_0\wedge\delta_1$, we can write
	\begin{align*}
	\frac{1}{T^2}\int_{0}^{T}\left\{\int_{0}^{y}f(T, y, s)\D s\right\}\D y
	& = \frac{1}{T^2}\int_{0}^{T}\left\{\int_{0}^{y}\left[f(T, y, 0) + \epsz(s)\right]\D s\right\}\D y\\
	& = \frac{1}{T^2}\int_{0}^{T}\left\{f(T, y, 0)y + 
	\int_{0}^{y}\epsz(s)\D s\right\}\D y\\
	& = \frac{1}{T^2}\int_{0}^{T}\left\{\Big(f(T, 0, 0)+\epst(y)\Big)y + 
	\int_{0}^{y}\epsz(s)\D s\right\}\D y\\
	& = \frac{f(T,0,0)}{2}
	+ \frac{1}{T^2}\int_{0}^{T}\left\{\epst(y)y + 
	\int_{0}^{y}\epsz(s)\D s\right\}\D y,
	\end{align*}
	where~$\epst(\cdot)$ is continuous on $[0,\delta_1]$ and~$\epsz(\cdot)$ is continuous on $[0,\delta_0]$, both null at the origin.
	For any $\eta_{1}>0$, there exists $\widetilde{\delta}_{1}>0$ such that, 
	for any $y \le 0, \widetilde{\delta}_{1}$, 
	$|\epst(y)|<\eta_{1}$.
	Therefore, for the first integral, we have, for $T<\widetilde{\delta}_{1}\wedge\delta_0\wedge\delta_1$,
	$$
	\left|\frac{1}{T^2}\int_{0}^{T}\epst(y)y\D y\right|
	\leq \frac{1}{T^2}\int_{0}^{T}\left|\epst(y)\right|y\D y
	\leq \frac{1}{T^2}\int_{0}^{T}\eta_1 y\D y
	\leq \frac{\eta_{1}}{2}.
	$$
	Likewise, since $\epsz(\cdot)$ tends to zero at the origin, 
	then for any $\eta_{0}>0$, there exists $\widetilde{\delta}_{0}>0$ such that, 
	for any $y \in [0, \widetilde{\delta}_{0}]$, 
	$|\epsz(y)|<\eta_{0}$.
	Therefore, for the second integral, we have, for $T<\widetilde{\delta}_{0}\wedge\delta_0\wedge\delta_1$,
	$$
	\left|\frac{1}{T^2}\int_{0}^{T}\int_{0}^{y}\epsz(s)\D s\D y\right|
	\leq \frac{1}{T^2}\int_{0}^{T}\int_{0}^{y}\left|\epsz(s)\right|\D s\D y
	\leq \frac{1}{T^2}\int_{0}^{T}\int_{0}^{y}\eta_{0} \D s \D y
	\leq \frac{\eta_{0}}{2}.
	$$
	Since $\eta_1$ and $\eta_0$ can be taken as small as desired, taking the limit of~$f(T,0,0)$ as~$T$ goes to zero concludes the proof.
\end{proof}

\begin{lemma}\label{lemma:MVT3}
	Let $f:\RR^4_+\to\RR$ be such that, for all~$0\le s\le y\le T$, $f(T,y,s,\cdot)$, $f(T,y,\cdot,0)$, $f(T,\cdot,0,0)$ are continuous on $[0,\delta_0]$, $[0,\delta_1]$, $[0,\delta_2]$ respectively
for some $\delta_0,\delta_1,\delta_2>0$, and~$\lim_{T\dto0}f(T,0,0,0)=f(0,0,0,0)$.
	Then the following limit holds:
	\begin{equation}\label{eq:MVT3}
	\lim_{T\downarrow0} \frac{1}{T^3} \int_0^T \int_0^y \int_0^s f(T,y,s,t) \dt \ds \dy = \frac{f(0,0,0,0)}{6}.
	\end{equation}
\end{lemma}\vspace{-10pt}
\begin{proof}
	For $T<\delta_0\wedge\delta_1\wedge\delta_2$, we can write \vspace{-6pt}
	\begin{align*}
	\frac{1}{T^3}&\int_{0}^{T}\left\{\int_{0}^{y} \left( \int_0^s f(T, y, s,t)\dt \right) \ds\right\}\dy\\
	&=\frac{1}{T^3}\int_{0}^{T}\left\{\int_{0}^{y} \left( \int_0^s [f(T, y, s,0)+\ep_0(t)]\dt \right) \ds\right\}\dy\\
	& = \frac{1}{T^3}\int_{0}^{T}\left\{\int_{0}^{y} \left( f(T, y, s,0)s + \int_0^s \ep_0(t)\dt \right) \ds\right\}\dy\\
	& = \frac{1}{T^3}\int_{0}^{T}\left\{\int_{0}^{y} \left( \big[f(T, y, 0,0)+\ep_1(s)\big]s + \int_0^s \ep_0(t)\dt \right) \ds\right\}\dy\\
	& = \frac{1}{T^3}\int_{0}^{T} \left\{f(T, y, 0,0) \frac{y^2}{2}+ \int_{0}^{y}\left( \ep_1(s)s + \int_0^s \ep_0(t)\dt \right) \ds\right\}\dy  \\
	& = \frac{1}{T^3}\int_{0}^{T} \left\{ \big[f(T, 0, 0,0)+\ep_2(y)\big] \frac{y^2}{2}+ \int_{0}^{y}\left( \ep_1(s)s + \int_0^s \ep_0(t)\dt \right) \ds\right\}\dy\\
	& = \frac{f(T,0,0,0)}{6} + \frac{1}{T^3}\int_{0}^{T} \left\{ \ep_2(y) \frac{y^2}{2} +\int_{0}^{y}\left( \ep_1(s)s + \int_0^s \ep_0(t)\dt \right) \ds\right\}\dy,
	\end{align*}
	where the function~$\ep_2$ is continuous on $[0,\delta_2]$, the function~$\epst$ is continuous on $[0,\delta_1]$ and the function~$\epsz$ is continuous on $[0,\delta_0]$, all converging to zero at the origin. By the same argument as in the previous proof, for any~$\eta_0,\eta_1,\eta_2>0$, there exists~$\widetilde\delta>0$ such that for all~$T\le\widetilde \delta$, $|\ep_0(T)|\le \eta_0$, $|\ep_1(T)|\le \eta_1$, and $|\ep_2(T)|\le \eta_2$. This entails
	$$
	\bigg\lvert\frac{1}{T^3}\int_{0}^{T} \left\{ \ep_2(y) \frac{y^2}{2} +\int_{0}^{y}\left( \ep_1(s)s + \int_0^s \ep_0(t)\dt \right) \ds\right\}\dy \bigg\lvert
	\le \frac{\eta_2+\eta_1+\eta_0}{6}.
	$$
	Since $\eta_2$, $\eta_1$ and $\eta_0$ can be taken as small as desired, taking the limit of~$f(T,0,0,0)$ as~$T$ goes to zero concludes the proof.
\end{proof}

To apply these lemmas, we will use a modified version of the martingale convergence theorem. It holds in our setting thanks to domination provided by {\COne} and {\CTwo} and the continuity of \Ffour.

\begin{lemma}\label{lemma:MCT} 
	Let~$(X_t)_{t\ge0}$  be almost surely continuous in a neighbourhood of zero, with 
	$\sup_{t\le 1} |X_t|\le Z\in L^1$. Then the conditional expectation process~$(\EE_t[X_t])_{t\ge0}$ is also almost surely continuous in a neighbourhood of zero. In particular,
	\[
	\lim_{t\downarrow0} \EE_t[X_t] =\EE[X_0].
	\]
\end{lemma}
\begin{remark}
	The process $(X_t)_{t\ge0}$ is not necessarily adapted. 
\end{remark}
\begin{proof}
	All the limits are taken in the almost sure sense.
	Let~$\delta>0$ be such that $X$ is continuous on $[0,\delta]$, and fix $t<\delta$. We set a sequence $\{t_n\}_{n\in\NN}$ on $[0,\delta]$ which converges to~$t$ as $n$ goes infinity.
	Assume first that $\{t_n\}_{n\in\NN}$ is a monotone sequence.
	Since $\Ff_{t_n}$ tends monotonically to~$\Ff_t$ and $X$ is dominated, the classical martingale convergence theorem (MCT) asserts that $\lim_{n\uparrow\infty} \EE_{t_n}[X_t] = \EE_t[X_t]$.
For fixed $n\in\NN$ and any~$\fq\ge |t_n-t|$,
	\begin{equation}\label{eq:MCTsup}
	|X_{t_n}-X_t| \le \sup_{|\fp-t|\le \fq} |X_{\fp} -X_t|. 
	\end{equation}
	Let us fix~$\ep>0$, by MCT, there exists~$n_0\in\NN$ such that, if~$n\ge n_0$ then 
	$$
	\bigg\lvert\EE_{t_n} \bigg[\sup_{|\fp-t|\le \fq} |X_{\fp} -X_t| \bigg]- \EE_t \bigg[\sup_{|\fp-t|\le \fq} |X_{\fp} -X_t| \bigg] \bigg\lvert < \ep,
	$$
	and by dominated convergence there exists~$\delta'>0$ with
	$\EE_t \Big[\sup_{|\fp-t|\le \delta'} |X_{\fp} -X_t| \Big]<\ep$. 
	There exists $n_1\in\NN$ such that $|t_n-t|\le\delta'$ for all $n\ge n_1$;
	thus if $n\ge n_0\vee n_1$,  ~\eqref{eq:MCTsup} yields $\EE_{t_n}[|X_{t_n}-X_t|]< 2\ep$ and
	\begin{equation}\label{eq:CondExpCty}
	\lim_{n\uparrow\infty} \EE_{t_n}[X_{t_n}] = \EE_t[X_t].
	\end{equation}
	Now we consider the general case where $\{t_n\}_{n\in\NN}$ is not monotone.
	From every subsequence of $\{t_n\}_{n\in\NN}$, one can extract a further subsequence which is monotone. Let us call this subsubsequence $\{t_{n_k}\}_{k\in\NN}$. Therefore, \eqref{eq:CondExpCty} holds with $t_{n_k}$ instead of $t_{n}$. Since every subsequence of $(\EE_{t_n}[X_{t_n}])_{n\in\NN}$ has a further subsequence that converges to the same limit, the original sequence also converges to this limit.
\end{proof}


For convenience, we use the following definition:
\begin{definition}
	Let $k,n\in\NN$ with~$k\le n$.
	For a function~$f:\RR_+^n\to\RR$, we denote
	$$
	\lim_{0\le x_1\le x_2\le\cdots\le x_k\downarrow0} f(x_1,\cdots,x_n):= \lim_{x_k\dto0} \cdots \lim_{x_2\dto0} \lim_{x_1\dto0} f(x_1,\cdots,x_n).
	$$
\end{definition}
Notice that the right-hand sides of~\eqref{eq:MVT1}, \eqref{eq:MVT2}, and \eqref{eq:MVT3} correspond to $\lim\limits_{y\le T\dto0} f(T,y)$,  $\displaystyle\half\lim\limits_{s\le y\le T\dto0} f(T,y,s)$ and 
$\displaystyle\frac{1}{6}\lim\limits_{t\le s\le y\le T\dto0} f(T,y,s,t)$ respectively.

\subsubsection{Proof of Proposition~\ref{prop:genmodel}}
Let us recall some important quantities:
\begin{align}
M_y &= \EE_y\left[\VIX_T\right] = \EE_y\left[ \sqrt{\frac1\Delta \int_T^{T+\Delta} \EE_T v_r \dr}\right],\nonumber\\ 
m^i_y &= \EE_y[\Df^i_y M_y]= \EE_y \left[\frac{\int_T^{T+\Delta} \Df^i_y \EE_T v_r \dr }{2\Delta \VIX_T}\right]
= \EE_y \left[\frac{\int_T^{T+\Delta} \Df^i_y v_r \dr }{2\Delta \VIX_T}\right], \nonumber\\
\phi_y^i&= \frac{m_y^i}{M_y} = \frac{\EE_y\left[\left(\int_T^{T+\Delta}\Df_y^i v_r \dr\right)/(2\Delta \VIX_T)\right]}{\EE_y[\VIX_T]}. 
\label{eq:phiVIX}
\end{align}
We also recall that $J_i$ and $G_{ij}$, $i,j\in \llbracket 1,N\rrbracket$ were defined in \eqref{eq:defJG}.
In this proof we will define $f(0):=\lim_{x\dto0}f(x)$, for every~$f:\RR_+\to\RR$, as soon as the limit exists and even if $f$ is not  actually continuous around zero. This way we make it continuous and it allows us to apply the convergence lemmas.\medskip

\textbf{Level.} By {\COne} and the martingale convergence theorem, 
$\lim_{y\downarrow0}\EE_y[\VIX_T]=\EE[\VIX_T]$ and $(M_y)_{y\ge0}$ is continuous around zero, almost surely. By {\Ffour} and the dominated convergence theorem (DCT), $\lim_{y\downarrow0} \int_T^{T+\Delta} \Df^i_y v_r\dr=\int_T^{T+\Delta} \Df^i_0 v_r\dr$ and $\big(\int_T^{T+\Delta} \Df^i_y v_r\dr)_{y\ge0}$ is continuous around zero, almost surely.
Let~$i\in\llbracket 1,N \rrbracket$, from {\COne} and {\CTwo} we also obtain that almost surely
$$
\frac{1}{\VIX_T}\int_T^{T+\Delta}\Df_y^i v_r \dr
\le X^2 \Big\{(T+\Delta-y)^{\Hp} - (T-y)^{\Hp}\Big\},
$$
for some~$X\in L^2$. 
Therefore it is dominated and by Lemma~\ref{lemma:MCT}, almost surely $m_y^i$ is continuous at zero and $\lim_{y\downarrow0} m^i_y = \EE\left[\frac{\int_T^{T+\Delta} \Df^i_0 v_r\dr}{2\Delta \VIX_T}\right]$.
Since~$M_y>0$ for all~$y\le T$, $\phi^i$ is also  continuous at zero and
$\lim_{y \le T\downarrow0} \phi^i_y = J_i / (2\Delta\VIX_0^2)$.
By virtue of Theorem~\ref{thm:level} and Lemma~\ref{lemma:MVT1}, we obtain
\[
\lim_{T\downarrow0} \Ii_{T} = \lim_{T\downarrow0} \EE[u_0] = \lim_{y\le T\downarrow0} \norm{\phib_y}= \frac{\norm{\Jb}}{2\Delta\VIX_0^2}.
\]
\textbf{Skew.}  To obtain the skew limit we need to compute a few Malliavin derivatives. For all~$i,j\in\llbracket 1,N \rrbracket$,
\begin{align*}
\Df^j_s m^i_y &= \EE_y \left[ \frac{\int_T^{T+\Delta}\Df^j_s \Df^i_y v_r \dr\, \VIX_T - \int_T^{T+\Delta} \Df^i_y v_r \dr \, \Df^j_s \VIX_T}{2\Delta \VIX_T^2}\right]\\
&= \EE_y\left[\frac{\int_T^{T+\Delta}\Df^j_s \Df^i_y v_r \dr}{2\Delta\VIX_T} - \frac{\int_T^{T+\Delta}\Df^i_y v_r \dr \int_T^{T+\Delta}\Df^j_s v_r \dr}{4\Delta^2\VIX_T^3} \right],
\end{align*}
which yields
\begin{align*}
\Df_s^j \phi_y^i = \frac{\Df^j_s m^i_y}{M_y} - \frac{ m^i_y \Df^j_s M_y}{M_y^2}
 & = \EE_y\left[\frac{\int_T^{T+\Delta}\Df^j_s \Df^i_y v_r \dr}{2\Delta \VIX_T M_y} - \frac{ \int_T^{T+\Delta}\Df_y^i v_r \dr \int_T^{T+\Delta}\Df_s^j v_r \dr}{4\Delta^2 \VIX^3_T M_y}  - \frac{ m_y^i \int_T^{T+\Delta} \Df^j_s v_r \dr}{2\Delta\VIX_T M_y^2}\right]\\
& =: \EE_y\left[A_T^{ij}(y,s) + B_T^{ij}(y,s) + C_T^{ij}(y,s)\right].
\end{align*}
Based on \COne, {\CTwo} and \Ffour, for each~$T\ge0$, $A^{ij}_T$, $B^{ij}_T$ and $C^{ij}_T$ are dominated and almost surely continuous in both arguments.  For each~$s\ge0$,
Lemma~\ref{lemma:MCT} and DCT yield, almost surely, that~$(\Df^j_s \phi^i_y)_{y\ge0}$ and~$(\Df^j_s \phi^i_0)_{s\ge0}$ are continuous around zero. In particular,
\begin{align*}
\lim_{s\downarrow0} \EE_y\big[ A_T^{ij}(y,s)+B_T^{ij}(y,s)+C_T^{ij}(y,s)\big] &= \EE\big[ A_T^{ij}(y,0)+B_T^{ij}(y,0)+C_T^{ij}(y,0)\big],\\
\lim_{y\downarrow0} \EE_y\big[ A_T^{ij}(y,0)+B_T^{ij}(y,0)+C_T^{ij}(y,0)\big] &= \EE\big[ A_T^{ij}(0,0)+B_T^{ij}(0,0)+C_T^{ij}(0,0)\big].
\end{align*}
By DCT again this yields
\begin{align*}
\lim_{T\downarrow0}\EE[ A_T^{ij}(0,0)] = \frac{G_{ij}}{2\Delta\VIX_0^2}
\qquad \text{and}\qquad
\lim_{T\downarrow0}\EE[ B_T^{ij}(0,0)] =\lim_{T\downarrow0} \EE[C_T^{ij}(0,0)]= -\frac{J_i J_j}{4\Delta^2 \VIX_0^4}.
\end{align*}
Therefore $\phi_s^j  \Df_s^j (\phi_y^i)^2$ satisfies the continuity requirements of $f(T,y,s)$ in Lemma~\ref{lemma:MVT2}. 
We combine this lemma with the limits above to see that, almost surely,
\begin{align*}
\lim_{T\downarrow0} \frac{1}{T^2} \int_0^T \phi_s^j \int_s^T \Df_s^j (\phi_y^i)^2 \dy \ds
&= \lim_{T\downarrow0} \frac{1}{T^2} \int_0^T  \int_0^y \phi_s^j \Df_s^j (\phi_y^i)^2 \ds \dy
= \half \lim_{s\le y \le T\dto0} \phi_s^j \Df_s^j (\phi_y^i)^2\\
&= \frac{J_j}{4\Delta \VIX_0^2} 
\left[ \frac{J_i G_{ij}}{2\Delta^2\VIX_0^4} - \frac{J_i^2 J_j}{2\Delta^3 \VIX_0^6} \right].
\end{align*}
We also recall that~$\lim_{T\downarrow0} u_0 =\frac{\norm{\Jb}}{2\Delta\VIX_0^2}$ almost surely, hence with {\CTwo} and \CThree, DCT entails
\begin{equation}\label{eq:skewVIX}
\lim_{T\downarrow0}  \Ss_T
=\sum_{i,j=1}^N\lim_{T\downarrow0}  \half \EE\left[\frac{\int_0^T \phi_s^j \int_s^T \Df_s^j (\phi_y^i)^2 \dy \ds}{u_0^3 T^{2}}\right]
= \frac{1}{2\norm{\Jb}^3}\sum_{i,j=1}^N J_i J_j \left(G_{ij} - \frac{J_i J_j}{\Delta\VIX_0^2}\right).
\end{equation}

\textbf{Curvature.}  We now turn our attention to the curvature. By the same arguments as above we have
\begin{align*}
\lim_{T\downarrow0}  \EE\left[ \frac{\left(\sum_{i,j=1}^N \int_0^T \phi_s^j \int_s^T \Df^j_s (\phi_y^i)^2 \dy \ds\right)^2}{u_0^7 T^4}\right]
&= \frac{2\Delta\VIX^2_0}{\norm{\Jb}^7}
\left(\sum_{i,j=1}^N J_i J_j \left(G_{ij} - \frac{J_i J_j}{\Delta\VIX_0^2}\right)\right)^2.
\end{align*}
For the last term of~\eqref{eq:curvature} we need to go one step further and compute more Malliavin derivatives since
\begin{align*}\label{eq:DTheta}
\Df^k_t \Theta^j_s &= \sum_{i=1}^N\left(\Df^k_t \phi_s^j \int_s^T \Df^j_s (\phi_y^i)^2\dy 
+2 \phi^j_s \int_s^T \Df^k_t \phi^i_y \Df^j_s \phi^i_y \dy + 2 \phi_s^j \int_s^T \phi^i_y \Df^k_t \Df^j_s \phi^i_y \dy\right) \\
& =: \sum_{i=1}^N \int_s^T \Upsilon^{ijk}(t,s,y,T)\dy. \nonumber
\end{align*}
Thus we zoom in on the last term of the display above,
\begin{align*}
\Df^k_t \Df^j_s \phi_y^i &= \frac{\Df^k_t \Df^j_s m^i_y \, M_y -\Df^j_s m^i_y \Df^k_t M_y}{M_y^2}- \frac{\Df^k_t m^i_y \Df^j_s M_y + m_y^i \Df^k_t \Df^j_s M_y}{M_y^2} + \frac{2 m^i_y \Df^j_s M_y \Df^k_t M_y}{M_y^3}\\
&=:\sum_{n=1}^5 Q_n^{ijk}(t,s,y,T),
\end{align*}
and zoom in again on $Q_1^{ijk}(t,s,y,T)$,
$$
\Df^k_t \Df^j_s m_y^i = \Df^k_t \Df^j_s \Df^i_y M_y 
= \Df^k_t \EE_y\left[ \frac{\int_T^{T+\Delta} \Df^j_s \Df^i_y v_r \dr}{2\Delta\VIX_T} - \frac{\int_T^{T+\Delta} \Df^i_y v_r\dr \int_T^{T+\Delta} \Df^j_s v_r\dr}{4\Delta^2\VIX_T^3}\right] 
=: \EE_y\left[ \alpha_T^{ijk}+\beta_T^{ijk}\right].
$$
Some additional computations lead to
\begin{align*}
\alpha_T^{ijk}&=  \frac{\VIX_T \int_T^{T+\Delta} \Df^k_t \Df^j_s\Df^i_y v_r\dr - \Df^k_t \VIX_T \, \int_T^{T+\Delta} \Df^j_s \Df^i_y v_r\dr}{2\Delta \VIX_T^2}\\
&=\frac{\int_T^{T+\Delta} \Df^k_t \Df^j_s\Df^i_y v_r\dr}{2\Delta \VIX_T} - \frac{\int_T^{T+\Delta} \Df^k_t v_r\dr \, \int_T^{T+\Delta} \Df^j_s\Df^i_y v_r\dr}{4\Delta^2\VIX_T^3};\\
\beta_T^{ijk} &=- \frac{\int_T^{T+\Delta} \Df^k_t \Df^i_y v_r\dr \, \int_T^{T+\Delta} \Df^j_s v_r\dr + \int_T^{T+\Delta} \Df^i_y v_r\dr \, \int_T^{T+\Delta} \Df^k_t \Df^j_s v_r\dr}{4\Delta^2 \VIX_T^3} \\
& \qquad\qquad + \frac{\Df^k_t \VIX_T^3 \int_T^{T+\Delta} \Df^i_y v_r\dr \, \int_T^{T+\Delta} \Df^j_s v_r\dr}{4\Delta^2 \VIX_T^6} \\
&=- \frac{\int_T^{T+\Delta} \Df^k_t \Df^i_y v_r\dr \, \int_T^{T+\Delta} \Df^j_s v_r\dr + \int_T^{T+\Delta} \Df^i_y v_r\dr \, \int_T^{T+\Delta} \Df^k_t \Df^j_s v_r\dr}{4\Delta^2 \VIX_T^3}\\
&\qquad\qquad + \frac{3   \int_T^{T+\Delta} \Df^k_t v_r\dr \,  \int_T^{T+\Delta} \Df^i_y v_r\dr \, \int_T^{T+\Delta} \Df^j_s v_r\dr}{8\Delta^3 \VIX_T^5} .
\end{align*}
We notice, crucially, that we have already justified the continuity of $\phi$ and $\Df\phi$ around zero in the proofs of level and skew respectively. Furthermore, by Lemma~\ref{lemma:MCT} the first two terms in~$\Upsilon^{ijk}$ as well as $Q_2,Q_3,Q_4,Q_5$ all converge to some finite limit as~$t\le s\le y\dto0$ and are continuous around zero, almost surely. Similarly, $\beta_T$ and the second term in $\alpha_T$ are almost surely continuous around zero, and their conditional expectation converges almost surely to some finite limit as~$t\le s\le y\dto0$ by DCT and Lemma~\ref{lemma:MCT}. Taking the limit $T$ to zero afterwards, all the aforementioned terms tend to a finite limit.
On the other hand, by \Ffour, DCT, and Lemma~\ref{lemma:MCT} we know that the conditional expectation of the first term in~$\alpha_T$ is almost surely continuous around zero, and its limit
\[
\lim_{t\le s\le y\downarrow0} \EE_y\left[\frac{\int_T^{T+\Delta} \Df_t^k \Df_s^j \Df^i_y v_r \dr}{2\Delta\VIX_T}\right] = \EE\left[\frac{\int_T^{T+\Delta} \Df_0^k \Df_0^j \Df^i_0 v_r \dr}{2\Delta\VIX_T}\right].
\]
Since $\gamma<0$, only this term contributes in the limit
\begin{align*}
\lim_{t\le s\le y\le T\downarrow0}  \frac{ \phi_t^k\,\Upsilon^{ijk}(t,s,y,T)}{T^{\gamma}}
&= \lim_{t\le s\le y\le T\downarrow0} 2 \phi_t^k \phi^j_s \phi^i_y \EE_y\left[\frac{\int_T^{T+\Delta} \Df_t^k \Df_s^j \Df^i_y v_r \dr}{2 T^\gamma\Delta\VIX_T M_y}\right] \\
&= \frac{J_i J_j J_k}{8\Delta^4\VIX_0^8} \, \lim_{T\dto0} \frac{\EE\left[\int_T^{T+\Delta} \Df_0^k \Df_0^j \Df^i_0 v_r \dr\right]}{T^\gamma},
\end{align*}
where we applied DCT at the end. Moreover, we know by {\CTwo} that this limit is finite for~$\gamma=3H-\half$, hence the conditions of Lemma~\ref{lemma:MVT3} are satisfied. We also recall that~$\lim_{T\downarrow0} u_0 =\frac{\norm{\Jb}}{2\Delta\VIX_0^2}$ almost surely, hence Lemma~\ref{lemma:MVT3} yields the almost sure limit
\begin{align*}
\lim_{T\downarrow0}\frac{1}{u_0^5 T^{3+\gamma}} \int_0^T \sum_{k=1}^N \left\{ \phi^k_t \Df^k_t \left( \int_t^T |\Thb_{s}|\ds \right) \right\} \dt 
& = \sum_{i,j,k=1}^N \lim_{T\downarrow0}\frac{1}{u_0^5 T^{3}}  \int_0^T  \int_0^y\int_0^s \frac{\phi^k_t\,\Upsilon^{ijk}(t,s,y,T)}{T^{3H-\half}} \dy\ds\dt \\
& = \frac{2 \Delta \VIX_0^{2}}{3\norm{\Jb}^5} \sum_{i,j,k=1}^N J_i J_j J_k \, \lim_{T\dto0} \frac{\EE\left[\int_T^{T+\Delta} \Df_0^k \Df_0^j \Df^i_0 v_r \dr\right]}{T^{3H-\half}} .
\end{align*}
The first two terms in~\eqref{eq:curvature} tend to zero since $\gamma<0$, hence Theorem~\ref{thm:curvature} and DCT
yield the final result
\begin{align*}\label{eq:curvatureVIX}
\lim_{T\downarrow0} \frac{\Cc_T}{T^{3H-\half}}
&=\frac{2\Delta \VIX_0^{2}}{3\norm{\Jb}^5} \sum_{i,j,k=1}^N J_i J_j J_k \lim_{T\dto0} \,\frac{\int_T^{T+\Delta}  \EE\left[\Df^k_0 \Df^j_0\Df^i_0 v_r\right]\dr}{T^{3H-\half}}. 
\end{align*}

\subsection{Proofs in the two-factor rough Bergomi model}

\subsubsection{Proof of Proposition~\ref{prop:suffconds}}\label{app:proofconds}
We start with a useful Lemma for Gaussian processes.
\begin{lemma}\label{lemma:supexpB}
If $B$ is a Gaussian process with 
$\norm{B}_T:=\sup\limits_{t\le T}|B_t|$,
then $\EE[\E^{p \|B\|_T}]$ is finite for all $p\in\RR$.
\end{lemma}

\begin{proof}
The Borell-TIS inequality asserts that~$\EE[\norm{B}_T]<\infty$ and~$\PP(\norm{B}_T-\EE\norm{B}_T>x)\le \exp\left\{-\frac{x^2}{2\sigma^2_T}\right\}$, where~$\sigma^2_T:=\sup_{t\le T} \EE[B_t^2]$, see~\cite[Theorem 2.1.1]{AT07}.  We then follow the proof of~\cite[Theorem 2.1.2]{AT07}.
	\begin{align*}
	\EE \left[\E^{p \norm{B}_T}\right] = \int_0^\infty \PP\left( \E^{p\norm{B}_T} > x\right) \dx
	\le \E^p + \EE[\norm{B}_T] + \int_{\E^p \vee \EE[\norm{B}_T]}^\infty \PP \left(\norm{B}_T>\frac{\log(x)}{p} \right) \dx.
	\end{align*}
	The Borell-TIS inequality in particular reads
	\[
	\PP \left(\norm{B}_T>\log(x^{1/p}) \right) \le 
	\exp\left\{-\frac{\left(\log(x^{1/p})-\EE[\norm{B}_T]\right)^2}{2\sigma_T^2}\right\}, 
	\qquad \text{for all  } u>\EE[\norm{B}_T].
	\]
	After a change of variable this yields
	\[
	\int_{\E^p \vee \EE[\norm{B}_T]}^\infty \PP \left(\norm{B}_T>\frac{\log(x)}{p} \right) \dx
	\le \int_{\frac{\log\left(\E^p \vee \EE[\norm{B}_T]\right)}{p}}^\infty 
	\exp\left\{-\frac{\left(x-\EE[\norm{B}_T]\right)^2}{2\sigma_T^2}\right\} p \E^{px} \dx <\infty,
	\]
	as desired.
\end{proof}
By the above lemma $\norm{v}_T\in L^p$, so that we can compute its Malliavin derivatives
\begin{align}\label{eq:Dv}
\Df^1_y v_r =  v_0 (r-y)^{\Hm} \Big(\chi \nu \Ee^1_r + \chib \eta \rho \Ee_r^2\Big)
\qquad\text{and}\qquad
\Df^2_y v_r = v_0 \chib \eta \rrho (r-y)^{\Hm} \Ee_r^2.
\end{align}
Without computing explicitly further derivatives, one notices that {\Ffour} holds 
and that there exist $C>0$ and a random variable~$X=C\norm{\Ee_r^1+\Ee_r^2}_T\in L^p$ for all~$p>1$ such that
$\Df^i_y v_r \le X (r-y)^{\Hm}$, $\Df^j_s \Df^i_y v_r \le X (r-s)^{\Hm} (r-y)^{\Hm}$
and
$\Df^k_t \Df^j_s \Df^i_y v_r \le X (r-t)^{\Hm} (r-s)^{\Hm} (r-y)^{\Hm}$, 
implying~\CTwo.
The following lemma grants \COne.
\begin{lemma}\label{lemma:oneoverM} 
	In the two-factor rough Bergomi model~\eqref{eq:expomodel} with~$0\leq T_1< T_2$,
	$$
	\EE\left[ \sup_{y\le T_1}\left(\EE_{y}\left[\frac{1}{T_2-T_1}\int_{T_1}^{T_2} v_{r} \dr\right]\right)^{-p}   \right]
	$$
	is finite for all $p>1$.
	In particular, $1/M$ is dominated in $L^p$.
\end{lemma} 

\begin{proof} 
	We first apply an $\exp-\log$ identity, and from the concave property of the logarithm function we may use Jensen's inequality, to obtain
$$
\EE_{y}\left[\frac{1}{T_2-T_1}\int_{T_1}^{T_2} v_{r} \dr\right]^{-p}
= \exp\left\{-p \log \left(\frac{1}{T_2-T_1}\int_{T_1}^{T_2} v_{r}\dr \right)\right\}
\le \exp\left\{-\frac{p}{T_2-T_1}\int_{T_1}^{T_2} \log(\EE_y [v_r]) \dr\right\}.
$$
We further bound $\log\EE_y [v_r]$ by the concavity of the logarithm and~\eqref{eq:expomodel}:
\begin{equation}\label{eq:lnvr}
-\log\EE_y [v_r] \le -\half \Big\{\log\left(2\chi v_0 \EE_y[\Ee^1_r]\right) + \log\left(2\chib v_0 \EE_y[\Ee^2_r]\right) \Big\},
\end{equation}
which we now compute as
	\begin{align} \label{eq:condexpexpo}
	\EE_y[\Ee^1_r]
& =  \exp\left\{-\frac{\nu^2 r^{2H}}{4H} + \nu \int_0^y (r-s)^{\Hm}\D W^1_s\right\} \EE_y\left[\exp\left\{\nu \int_y^r (r-s)^{\Hm}\D W^1_s\right\}\right] \nonumber\\
& = \exp\left\{\frac{\nu^2}{4H} \left[(r-y)^{2H}-r^{2H}\right] + \nu \int_{0}^{y} (r-s)^{\Hm} \D W^1_s\right\}; \nonumber \\
	\EE_y[\Ee^2_r] &=\exp\left\{\frac{\eta^2}{4H}\left[(r-y)^{2H}-r^{2H}\right] + \eta \int_0^{y} (r-s)^{\Hm} \D (\rho W^1_s+\rrho W^2_s)\right\}.
\end{align}
	Let us deal with the first term of~\eqref{eq:lnvr}, as the second one is analogous. We have
	\[
	\int_{T_1}^{T_2} \left[(r-y)^{2H}-r^{2H}\right] \dr = \frac{(T_2-u)^{2\Hp}-(T_1-y)^{2\Hp}-T_2^{2\Hp}+T_1^{2\Hp}}{2\Hp} \vspace{-.2cm}
	\]
	is clearly bounded below for all~$0\le u\le T_1$. Moreover, by Fubini's theorem,
$$
\int_{T_1}^{T_2}  \int_0^y (r-t)^{\Hm} \D W^1_t \dr 
=\int_0^y \int_{T_1}^{T_2}   (r-t)^{\Hm} \dr\D W^1_t 
= \int_0^y \frac{(T_2-t)^{\Hp}-(T_1-t)^{\Hp}}{\Hp} \D W^1_t 
=: \overline{B}_t
$$
is a Gaussian process. Since $\exp\{\cdot\}$ is increasing, $\sup_{t\in[0,T]}  \exp\{\overline{B}_t\} = \exp\{\sup_{t\in[0,T]} \overline{B}_t\}$, thus
	\[
	\EE\left[\sup_{y\le T_1} \exp\left(-\frac{p}{T_2-T_1}\int_{T_1}^{T_2}  \int_0^y (r-s)^{\Hm} \D W^1_s\dr\right)\right] \le \EE\left[\exp\left(\frac{p}{T_2-T_1} \norm{\overline{B}}_T\right)\right] <\infty,
	\]
	by Lemma~\ref{lemma:supexpB}, which concludes the proof.
\end{proof}



Combining~\eqref{eq:Dv}, \eqref{eq:condexpexpo} we obtain
$\EE_y[\Df^i_y v_r]$, $i=1,2$.
The following lemma proves that {\CThree} is satisfied.
\begin{lemma}
	For any~$p>1$,
	$ \EE[u_s^{-p}]$ is uniformly bounded in~$s$ and~$T$, with~$s\le T$.
\end{lemma}
\begin{proof}
	Since $\nu,\eta,\rho+\rrho>0$, then $\Df^1_y v_r+\Df^2_y v_r>0$ almost surely for all~$y\le r$. Moreover, $\VIX$ and~$1/\VIX$ are dominated by some~$X \in L^p$ for all~$p>1$, then almost surely and independently of the sign of the numerator, we obtain
	\[
	m^i_y = \EE_y\left[\frac{\int_T^{T+\Delta} \Df^i_y v_r \dr}{2\Delta \VIX_T}\right] \ge \EE_y\left[\frac{\int_T^{T+\Delta} \Df^i_y v_r \dr}{2\Delta X}\right],
	\]
	and therefore, using that~$1/M$ is dominated by~$X$ and Jensen's inequality we get
	\begin{align}\label{eq:invucarre}
	\frac{1}{u_s^2}
	= \frac{T-s}{\int_s^T \sum_{i=1}^N (\phi^i_y)^2\dy}
	& \le  \frac{X^2(T-s)}{\int_s^T \sum_{i=1}^N (m^i_y)^2\dy} 
	\le \frac{X^2 N(T-s)}{\int_s^T \left(\sum_{i=1}^N m^i_y \right)^2\dy} 
	\le X^2 N \left( \frac{T-s}{\int_s^T \sum_{i=1}^N m^i_y\dy}\right)^2 \nonumber\\
	&\le 4 X^2 N \left(\frac{\int_s^T\int_T^{T+\Delta}\sum_{i=1}^N \EE_y\left[\Df^i_y v_r /X\right] \dr\dy}{\Delta(T-s)}\right)^{-2}.
	\end{align}
	Hence we turn our attention to
	\begin{align}
	&\EE\left[ \left(\frac{1}{\Delta(T-s)} \int_s^T \int_T^{T+\Delta} \EE_y\left[\frac{\Df^1_y v_r+\Df^2_y v_r}{X}\right] \dr \dy\right)^{-p}\right] \nonumber\\
	&= \EE\left[ \exp\left\{-p\log \left(\frac{1}{\Delta(T-s)} \int_s^T \int_T^{T+\Delta} \EE_y\left[\frac{\Df^1_y v_r+\Df^2_y v_r}{X}\right]  \dr \dy\right)\right\}\right] \nonumber  \\
	&\le \EE\left[ \exp\left\{-\frac{p}{\Delta(T-s)} \int_s^T \int_T^{T+\Delta} \EE_y\left[\log \left(\Df^1_y v_r+\Df^2_y v_r\right) -\log(X)\right] \dr \dy \right\}\right] \nonumber\\
	&\le \left(\EE\left[ \exp\left\{-\frac{2p}{\Delta(T-s)} \int_s^T \int_T^{T+\Delta} \EE_y\left[\log \left(\Df^1_y v_r+\Df^2_y v_r \right)\right] \dr \dy \right\}\right]\right)^\half \sqrt{\EE\left[X^{2p}\right]},
	\label{eq:expmomentw}
	\end{align}
	using Jensen's and Cauchy-Schwarz inequalities and~$\E^{p\EE_y[\log(X)]}\le \EE_y[X^p]$. 
Convexity and~\eqref{eq:Dv} imply
	\begin{align*}
	-\log \left(\Df^1_y v_r+\Df^2_y v_r\right) 
	\le &-\frac12 \Big\{\log \left(2 v_0\chi \nu(r-y)^{\Hm} \Ee^1_r \right)
	+ \log\left(2 v_0 \chib \eta(\rho+\rrho)(r-y)^{\Hm} \Ee_r^2\right) \Big\}.
	\end{align*}
	We focus on the first term and the other can be treated identically. From~\eqref{eq:condexpexpo} we have
	\begin{align}\label{eq:lnw}
	\EE_y\left[\log \left(2 v_0 \chi \nu(r-y)^{\Hm} \Ee^1_r \right)\right]
	= \log\left(2 v_0 \chi \nu(r-y)^{\Hm}\right) - \frac{\nu^2 r^{2H}}{4H} + \nu \int_0^y (r-t)^{\Hm} \D W^1_t .
	\end{align}
	Let us start with
	\begin{align*}
	&\int_s^T \int_T^{T+\Delta} \log\left(2 v_0 \chi \nu(r-y)^{\Hm}\right) \dr \dy \\
	&= 2(T-s)\Delta v_0 \chi \nu + \Hm \int_s^T \int_T^{T+\Delta} \log(r-y) \dr \dy  \\
	&= 2(T-s)\Delta v_0 \chi \nu + \Hm \int_s^T \Big[(T+\Delta-y)\log(T+\Delta-y) - (T+\Delta-y)
- (T-y)\log(T-y) + (T-y) \Big]\dy \\
	& = 2(T-s)\Delta v_0 \chi \nu + \Hm \left\{-\Delta (T-s) - \int_\Delta^{T+\Delta-s} x\log(x)\dx +\int_0^{T-s} x\log(x)\dx \right\} \\
	& = 2(T-s)\Delta v_0 \chi \nu + \Hm \Bigg\{-\Delta (T-s)+ \left(\frac{(T-s)^2 \log(T-s)}{2} - \frac{(T-s)^2}{4}\right)  \\
	& \qquad\qquad\quad - \left(\frac{(T+\Delta-s)^2 \log(T+\Delta-s)}{2} - \frac{(T+\Delta-s)^2}{4}\frac{\Delta^2\log(\Delta)}{2} +\frac{\Delta^2}{4}\right) \Bigg\}\\
	&= 2(T-s)\Delta v_0 \chi \nu +(T-s) \Hm \Bigg\{ -\Delta + \left(\frac{(T-s) \log(T-s)}{2} - \frac{(T-s)}{4}\right) + \frac{2\Delta +(T-s)}{4}\\
	& \quad \qquad\qquad\qquad - \Delta\log(T+\Delta-s) - \frac{(T-s)\log(T+\Delta-s)}{2} \Bigg\} -\frac{\Delta^2}{2} \big(\log(T+\Delta-s)-\log(\Delta) \big).
	\end{align*}
	By Taylor's theorem $\log(T+\Delta-s)-\log(\Delta) = \frac{T-s}{\Delta} + \ep(T-s)$, where~$\ep:\RR_+\to\RR_+$ is such that~$\ep(x)/x$ tends to zero at the origin. We conclude that
	\[
	-\frac{p}{2\Delta(T-s)} \int_s^T \int_T^{T+\Delta} \log\left(2 v_0 \chi \nu(r-y)^{\Hm}\right) \dr \dy
	\]
	is uniformly bounded.
	Now we study the second term of~\eqref{eq:lnw}
	\[
	- \int_s^T \int_T^{T+\Delta} r^{2H} \dr\dy = (T-s) \frac{T^{2\Hp}-(T+\Delta)^{2\Hp}}{2\Hp}.
	\]
	Therefore the following is uniformly bounded:
	\[
	\frac{p}{2\Delta(T-s)} \int_s^T \int_T^{T+\Delta} \frac{\nu^2 r^{2H}}{4H}  \dr \dy.
	\]
	And for the last term we get by stochastic Fubini's theorem~\cite[Theorem 65]{Protter05}
	\begin{align*}
	\int_s^T \int_T^{T+\Delta} \int_0^y (r-t)^{\Hm} \D W^1_t \dr\dy
	& = \int_s^T \int_0^y \int_T^{T+\Delta} (r-t)^{\Hm} \dr \D W^1_t \dy\\
	& =\int_0^T \int_{s\vee t}^T \frac{(T+\Delta-t)^{\Hp} - (T-t)^{\Hp}}{\Hp} \dy \D W^1_t.
	\end{align*}
	Standard Gaussian computations then yield
	\begin{align}
	&\EE\left[\exp\left\{-\frac{p}{4\Delta(T-s)} \int_s^T \int_T^{T+\Delta} \nu \int_0^y (r-t)^{\Hm} \D W^1_t \dr\dy \right\}\right] \nonumber \\
	&= \exp\left\{\half \left(\frac{p\nu}{4\Delta(T-s)}\right)^2 \int_0^T \left( \int_{s\vee t}^T \frac{(T+\Delta-t)^{\Hp} - (T-t)^{\Hp}}{\Hp} \dy\right)^2\dt \right\}.
	\label{eq:TripleInt}
	\end{align}
	The incremental function~$x\mapsto (x+\Delta)^{\Hp}-x^{\Hp}$ is decreasing by concavity, hence~$(T+\Delta-t)^{\Hp} - (T-t)^{\Hp} \le \Delta^{\Hp}$ and we obtain
	\[
	\int_0^T (T-s\vee t)^2 \dt = \int_0^s (T-s)^2\dt + \int_s^T (T-t)^2\dt = s(T-s)^2 + \frac{(T-s)^3}{3},
	\]
	which entails that~\eqref{eq:TripleInt} is uniformly bounded. We thus showed that~\eqref{eq:expmomentw} is uniformly bounded in~$s,T$. 
	
	Coming back to~\eqref{eq:invucarre} we have by Cauchy-Schwarz inequality
	\[
	\EE[u_s^{-p}]^2 \le 2^{p} \EE[X^{p}] \, \EE\left[ \left(\frac{T-s}{\int_s^T (m^1_y + m^2_y) \dy} \right)^{2p} \right],
	\]
	which is uniformly bounded for all~$s\le T$, and this concludes the proof.
\end{proof}
\subsubsection{Proof of Proposition~\ref{prop:expomodel}}\label{sec:proofexpomodel}

\textbf{Level.} We start with the derivatives
\begin{align*}
\Df^1_s v_t =  v_0\Big[\chi \nu(t-s)^{\Hm} \Ee^1_t + \chib \eta \rho (t-s)^{\Hm} \Ee_t^2\Big]
\qquad\text{and}\qquad
\Df^2_s v_t = v_0 \chib \eta \rrho (t-s)^{\Hm} \Ee_t^2,
\end{align*}
and recall the definitions~\eqref{eq:defJG}
$$
J_1 = \int_0^\Delta v_0  \EE\left[\chi \nu r^{\Hm} \Ee^1_r +\chib\eta \rho r^{\Hm} \Ee^2_r \right] \dr = v_0(\chi\nu+\chib\eta\rho)\frac{\Delta^{\Hp}}{\Hp}
\qquad\text{and}\qquad
J_2 = v_0\chib \eta \rrho \frac{\Delta^{\Hp}}{\Hp}.
$$
We also note that~$\EE[\Ee^i_t]=1$. This yields the norm
\[
\norm{\Jb} := \left(J_1^2 + J_2^2 \right)^{\half}
= \frac{v_0 \Delta^{\Hp}}{\Hp} 
\sqrt{ (\chi\nu+\chib\eta\rho)^2 + \chib^2\eta^2\rrho^2}
= \frac{v_0 \Delta^{\Hp}}{\Hp} \psi(\rho,\nu,\eta,\chi),
\]
with the function~$\psi$ defined in the proposition,
which grants us the first limit by Proposition~\ref{prop:genmodel}.
To simplify the notations below, we introduce~$\wf := \chi\nu + \chib\eta\rho$.

\textbf{Skew.} We compute the further derivatives:
\begin{align*}
\Df^1_0 \Df^1_0 v_t &=  v_0\left(\chi\nu^2 t^{2H-1} \Ee_t^1 + \chib\eta^2 \rho^2 t^{2H-1} \Ee_t^2 \right),\\
\Df^1_0 \Df^2_0 v_t & =  v_0\chib\eta^2 \rho \rrho t^{2H-1} \Ee^2_t,\\
\Df^2_0 \Df^2_0 v_t & =  v_0\chib \eta^2 \rrho^2 t^{2H-1} \Ee^2_t.
\end{align*}
Similarly to~$J$, we recall that~$G_{ij}= \int_0^{\Delta} \EE\big[\Df^j_0 \Df^i_0 v_r \big]\dr$, such that
\begin{align*}
G_{11}=\frac{\Delta^{2H}}{2H} v_0 (\chi\nu^2 + \chib\eta^2 \rho^2), \qquad
G_{12} = \frac{\Delta^{2H}}{2H} v_0\chib \eta^2\rho \rrho, \qquad
G_{22}= \frac{\Delta^{2H}}{2H} v_0\chib \eta^2 \rrho^2.
\end{align*}
Notice that~$\VIX^2_0=v_0$, thus we have
\begin{align*}
J_1^2 \left(G_{11} -\frac{ J_1^2}{\Delta\VIX_0^2} \right)
&= \frac{v_0^3\Delta^{4H+1}}{2H\Hp^2} 
\wf^2 (\chi\nu^2+\chib\eta^2\rho^2)
- \frac{v_0^3 \Delta^{4H+1}}{\Hp^4}  \wf^4\\
&= v_0^3 \frac{\Delta^{4H+1}}{\Hp^2}  \wf^2
\left[ \frac{\chi\nu^2+\chib\eta^2\rho^2}{2H}-\frac{\wf^2}{\Hp^2}\right]; \\
J_1 J_2 \left( G_{12} -\frac{J_1 J_2}{\Delta\VIX_0^2} \right)
&= \frac{v_0^3\Delta^{4H+1}}{2H\Hp^2}\wf\chib^2\eta^3\rho\rrho^2
-  \frac{v_0^3\Delta^{4H+1}}{\Hp^4}  \wf^2\chib^2\eta^2\rrho^2 \\
&= v_0^3 \frac{\Delta^{4H+1}}{\Hp^2} \wf\chib^2\eta^2\rrho^2 \left[ \frac{\eta\rho}{2H} - \frac{\wf}{\Hp^2}\right];\\
J_2^2 \left( G_{22} -\frac{J_2^2}{\Delta\VIX_0^2} \right)
&= v_0^3 \frac{\Delta^{4H+1}}{2H\Hp^2} \chib^3 \eta^4\rrho^4 - v_0^3 \chib^4 \frac{\Delta^{4H+1}}{\Hp^4} \eta^4\rrho^4 \\
&= v_0^3 \frac{\Delta^{4H+1}}{\Hp^2}\chib^3 \eta^4\rrho^4 \left(\frac{1}{2H}-\frac{\chib}{\Hp^2}\right). \end{align*}
Finally by Proposition~\ref{prop:genmodel} we obtain
\begin{align*}
\lim_{T\downarrow0} \Ss_T =
& \frac{\Hp\Delta^{\Hm}}{2\psi(\rho,\nu,\eta,\chi)^{3}} \Bigg\{\wf^2 \left[ \frac{\chi\nu^2+\chib\eta^2\rho^2}{2H}-\left(\frac{\wf}{\Hp}\right)^2\right]\\
& + 2 \wf\chib^2\eta^2\rrho^2 \left[ \frac{\eta\rho}{2H} - \frac{\nu+\eta\rho}{\Hp^2}\right] 
+  \chib^3 \eta^4\rrho^4 \left(\frac{1}{2H}-\frac{1}{\Hp^2}\right) \Bigg\}.
\end{align*}

\textbf{Curvature.} For the last step we go one step further:
\begin{alignat*}{2}
\Df^1_0\Df^1_0 \Df^1_0 v_t &=  v_0 \left(\chi\nu^3 \Ee_t^1 + \chib\eta^3 \rho^3 \Ee_t^2 \right)t^{3\Hm}; \qquad\qquad
&&\Df^2_0 \Df^1_0 \Df^1_0 v_t =  v_0\chib \eta^3 \rho^2 \rrho t^{3\Hm} \Ee_t^2; \\
\Df^2_0 \Df^2_0 \Df^1_0 v_t &=  v_0\chib \eta^3 \rho \rrho^2 t^{3\Hm} \Ee_t^2;  \qquad\qquad
&&\Df^2_0 \Df^2_0 \Df^2_0 v_t =  v_0\chib \eta^3  \rrho^3 t^{3\Hm} \Ee_t^2.
\end{alignat*}
We notice that
\[
\lim_{T\downarrow0} \frac{\int_T^{T+\Delta} r^{3\Hm} \dr}{T^{3H-\half}}
 = \lim_{T\downarrow0} \frac{(T+\Delta)^{3H-\half} - T^{3H-\half}}{T^{3H-\half} (3H-\half)}
 = \frac{2}{1-6H}.
\]
By the curvature limit in Proposition~\ref{prop:genmodel}, we have
\begin{alignat*}{2}
&\lim_{T\downarrow0} \frac{\Cc_T}{T^{3H-\half}}
&&= \frac{2\Delta v_0}{3\left(\frac{v_0 \Delta^{\Hp}}{2\Hp}\right)^5 \psi(\rho,\nu,\eta,\chi)^{5}} \Bigg\{ 
v_0^3\frac{\Delta^{3\Hp}}{\Hp^3} \wf^3 v_0 \frac{\chi\nu^3+\chib\eta^3\rho^3}{\half-3H} 
\\
& + 3v_0^3\frac{\Delta^{3\Hp}}{\Hp^3}
&& \wf^2 \chib\eta\rrho v_0 \frac{\chib\eta^3\rho^2\rrho}{\half-3H}
+ 3v_0^3\frac{\Delta^{3\Hp}}{\Hp^3} \wf \chib^2\eta^2\rrho^2 v_0\frac{\chib\eta^3\rho\rrho^2}{\half-3H} 
+ v_0\frac{\Delta^{3\Hp}}{\Hp^3} \chib^3 \eta^3\rrho^3 v_0 \frac{\chib\eta^3\rrho^3}{\half-3H} \Bigg\} \\
& &&= \frac{128 \Delta^{-2H} \Hp^2}{3\psi(\rho,\nu,\eta,\chi)^{5}(1-6H)}
\Big\{ \wf^3 (\chi\nu^3+\chib\eta^3\rho^3) 
+ 3 \wf^2 \chib^2 \eta^4\rrho^2 \rho^2
+ 3  \wf\chib^3  \eta^5\rrho^4 \rho 
+  \chib^4\eta^6\rrho^6 \Big\},
\end{alignat*}
which yields the claim. \qedhere

\subsection{Proofs in the stock price case}\label{app:SPXlimit}

\subsubsection{Proof of Proposition~\ref{prop:SPXlimit}}
Since~$\phi$ and $u^{-p}$ are dominated by conditions (i) and (iii) respectively, with the same notations as in the proof of Proposition~\ref{prop:condsaresufficient}, we obtain by (ii), as~$T$ goes to zero,
\begin{align*}
\Df \phi_s  \lesssim (T-s)^{\Hm}, \qquad
\Thb_s  \lesssim (T-s)^{\Hp}, \qquad
\Df \Thb_s  \lesssim (T-s)^{2H}, \qquad
\Df \Df \Thb_s  \lesssim (T-s)^{3H-\half}.
\end{align*} 
Under our three assumptions it is straightforward to see that {\HAll} are satisfied. 
Moreover, the terms in {\HSix} behave as~$T^{2H-\lambda}$ and the one {\HSeven} as~$T^{\Hm-\lambda}$, which means that by setting~$\lambda=\Hm$ the former vanishes and the second yields a non-trivial behaviour.

Let us have a look at the short-time implied volatility. 
By Lemma~\ref{lemma:MVT1} and the continuity of $v$ we have~$\lim_{T\downarrow0} u_0 = \sqrt{\sum_{i=1}^N v_0 \rho_i^2}=\sqrt{v_0}$  almost surely, hence by Theorem~\ref{thm:level} and dominated convergence
\[
\lim_{T\downarrow0}\widehat \Ii_{T} = \lim_{T\downarrow0} \EE[u_0] = \sqrt{v_0}.
\]
We then turn our attention to the short-time skew. 
With $\lambda=\Hm$,
Theorem~\ref{thm:skew} and DCT imply
\begin{align*}
\lim_{T\downarrow0} \frac{\widehat\Ss_T}{T^{\Hm}}
=\sum_{i,j=1}^N\lim_{T\downarrow0}  \half \EE\left[\frac{\int_0^T \phi_s^j \int_s^T \Df_s^j (\phi_y^i)^2 \dy \ds}{u_0^3 T^{\frac32+H}}\right]
&= \sum_{j=1}^N \frac{\rho_j}{2 v_0^{3/2}} \EE\left[ \lim_{T\downarrow0} \frac{\int_0^T \int_s^T \sqrt{v_s}\,  \Df^j_{s} v_y\dy\ds}{T^{\frac32+H}} \right],
\end{align*}
where we used~$\sum_{i=1}^N \rho_i^2=1$. 
For any~$j\in\llbracket 1,N\rrbracket$, Cauchy-Schwarz inequality yields
\[
\EE\left[ \left( \sqrt{\frac{v_s}{v_0}}-1\right) \Df^j_s v_y \right] \le \EE\left[ \left( \sqrt{\frac{v_s}{v_0}}-1\right)^2\right]^\half \, \EE\big[ (\Df^j_s v_y)^2 \big]^\half,
\]
where $\EE\big[ (\Df^j_s v_y)^2 \big]^\half\le C (y-s)^{\Hm}$ for some finite constant~$C$ by (ii). Therefore, 
\begin{align*}
\lim_{T\downarrow0} \bigg(&\frac{\int_0^T \int_s^T \EE[ \sqrt{v_s}\,  \Df^j_{s} v_y]\dy\ds}{\sqrt{v_0}\,T^{\frac32+H}} - \frac{\int_0^T \int_s^T \EE[  \Df^j_{s} v_y]\dy\ds}{T^{\frac32+H}} \bigg) \\
&\qquad\qquad\le C \lim_{T\dto0}\left( \sup_{t\le T} \EE\left[ \left( \sqrt{\frac{v_t}{v_0}}-1\right)^2\right]^\half \frac{ \int_0^T \int_s^T (y-s)^{\Hm}\dy\ds}{T^{\frac32+H}}\right).
\end{align*}
Since the fraction is equal to $((H+\frac32)\Hp)^{-1}$ and $\limsup_{T\dto0} \EE\big[(\sqrt{v_t/v_0}-1)^2 \big]$ is null by~(iv), we obtain
\[
\lim_{T\downarrow0} \frac{\widehat\Ss_T}{T^{\Hm}}
= \sum_{j=1}^N \frac{\rho_j}{2 v_0} \EE\left[ \lim_{T\downarrow0} \frac{\int_0^T \int_s^T   \Df^j_{s} v_y\dy\ds}{T^{\frac32+H}} \right]. 
\]

\subsubsection{Proof of Corollary~\ref{coro:SPXlimit}}
Since $\EE[u_s^{-p}] = \EE\left[\left(\frac{1}{T-s}\int_s^T v_r\dr \right)^{-\frac{p}{2}}\right]$, 
Lemmas~\ref{lemma:supexpB} and~\ref{lemma:oneoverM} show that assumptions~(i)-(iii) of Proposition~\ref{prop:SPXlimit} hold. 
Moreover~$v$ has almost sure continuous paths, hence $\sqrt{\frac{v_t}{v_0}}$ tends to one almost surely 
and~(iv) holds by reverse Fatou's lemma.
For $0\le s\le y$, \eqref{eq:Dv} implies
$$
\EE[\Df^1_s v_y] = v_0 (y-s)^{\Hm} \big(\chi \nu +\chib\eta\rho\big)
\qquad\text{and}\qquad
\EE[\Df^2_s v_y] = v_0 (y-s)^{\Hm}\chib\eta\rrho,
$$
and clearly $\EE[\Df^3_s v_y] =0$.
Therefore, Proposition \ref{prop:SPXlimit} entails
\[
\lim_{T\downarrow0} \frac{\widehat\Ss_T}{T^{\Hm}} = \frac{\rho_1}{2v_0}\frac{v_0(\chi\nu+\chib\eta\rho)}{\Hp(H+\frac32)} + \frac{\rho_2}{ 2v_0}\frac{v_0\chib\eta\rrho}{\Hp(H+\frac{3}{2})} 
= \frac{\rho_1 \chi\nu + \eta \chib (\rho_1 \rho+\rho_2\rrho)}{(2\Hp)(H+\frac32)}. 
\]


\subsection{Partial derivatives of the Black-Scholes function} \label{app:partiald}
Recall the Black-Scholes formula from~\eqref{eq:BSFormula} and assume that $\varsigma:=\sigma\sqrt{T-t}>0$ is fixed.
Then
$$
\pdx \BS(t,x,k,\sigma) = \E^x \Nn(d_+(x,k,\sigma))
\quad\text{and}\quad
\pdx^2 \BS(t,x,k,\sigma) = \E^x \left\{\Nn(d_+(x,k,\sigma)) + \frac{\Nn'(d_+(x,k,\sigma))}{\varsigma}\right\},
$$
such that (we drop the dependence on $t$ and $\sigma$ in the $G(\cdot)$ notation)
$$
G(x,k) := (\pdx^2 - \pdx) \BS(t,x,k,\sigma) = \frac{\E^{x -\half d_+(x,k,\sigma)^2}}{\varsigma\sqrt{2\pi}} 
= \frac{\E^{k-\half d_-(x,k,\sigma)^2}}{\varsigma\sqrt{2\pi}}.
$$
Define now
$$
f(x,k) := x- \frac{d_+(x,k,\sigma)^2}{2} = k - \frac{d_-(x,k,\sigma)^2}{2} = \frac{x+k}{2} - \frac{(x-k)^2}{2\varsigma^2} -\frac{\varsigma^2}{8}.
$$
We then have
\begin{align*}
\pdx f(x,k) & = \half - \frac{x-k}{\varsigma^2},\qquad
\pdk f(x,k) = \half+ \frac{x-k}{\varsigma^2},\\
\pdx^2 f(x,k) & = \pdk^2 f(x,k) = -\pd^2_{xk} f(x,k) = - \frac{1}{\varsigma^2}.
\end{align*}
For the partial derivatives, noting that~$\pdx G = \frac{1}{\varsigma\sqrt{2\pi}} \pdx f \E^{f}$ implies the ATM formula
\begin{align*}
\pdx G(x,x) =  \frac{1}{2\varsigma\sqrt{2\pi}}\exp\left\{x-\frac{\varsigma^2}{8}\right\},
\end{align*}
and furthermore, 
$$
\pdxk G = \frac{\E^f}{\varsigma\sqrt{2\pi}}  \Big( \pdxk f + \pdx f \pdk f \Big)
\qquad\text{and}\qquad
\pdxk G(x,x) = \frac{\E^{f(x,x)}}{\varsigma\sqrt{2\pi}} \left( \frac{1}{\varsigma^2} + \frac14 \right).
$$
We further define the partial derivatives appearing in the proof of Theorem~\ref{thm:skew}, after~\eqref{eq:SkewExpressionII},
\begin{align*}
L(x,k) &:= \left(\frac14 \partial_x + \half \partial_{xk}\right)G(x,k)
 = \frac{1}{\varsigma\sqrt{2\pi}} \E^{f(x,k)} \left(\frac14 + \frac14 \partial_k f(x,k) - \half (\partial_k f(x,k))^2 - \half \partial_{kk} f(x,k) \right),\\
L(x,x)&= \frac{\E^{f(x,x)}}{\varsigma\sqrt{2\pi}} \left( \frac14 + \frac{1}{2\varsigma^2} \right).&
\end{align*}
Using~$\partial_k f = 1-\partial_x f$ and~$\partial_{xk} f = - \partial_{xx} f = -\partial_{kk}f$ we compute
\begin{align*}
\partial_k L &=  \frac{\E^f}{\varsigma\sqrt{2\pi}} \left[ \frac34 \partial_x f - \frac54 (\partial_x f)^2+\half (\partial_x f)^3 - \frac54 \partial_{xx} f + \frac32 \partial_x f \partial_{xx}f \right],\\
\partial_k L(x,x) &= \frac{\E^{f(x,x)}}{\varsigma\sqrt{2\pi}}  \left( \frac18 + \frac{1}{2\varsigma^2}\right).
\end{align*}

Finally, we need the derivatives featuring in the proof of Theorem~\ref{thm:curvature}. We start with
\begin{align*}
\widetilde H = \partial_{xk} L &=  \frac{\E^f}{\varsigma\sqrt{2\pi}} \left[ \frac34 (\dxf)^2 -\frac54 (\dxf)^3 + \half (\dxf)^4 + \frac34 \dxxf - \frac{15}{4} \dxf \dxxf + 3 (\dxf)^2 \dxxf + \frac32 (\dxxf)^2 \right]\\
\partial_{xk} L(x,x) &= \frac{\E^{f(x,x)}}{\varsigma\sqrt{2\pi}} \left(\frac{1}{16} + \frac{3}{8\varsigma^2} + \frac{3}{2\varsigma^4}\right).
\end{align*}
The next partial derivative yields
\begin{align*}
\partial_{xxk} L & =  \frac{\E^f}{\varsigma\sqrt{2\pi}} \left[ \frac34 (\dxf)^3 - \frac54 (\dxf)^4 + \half (\dxf)^5 + \frac94 \dxxf \dxf - \frac{15}{2} \dxxf (\dxf)^2 \right.\\
& \qquad\qquad \left.- \frac{15}{4} (\dxxf)^2 + 5\dxxf (\dxf)^3 + \frac{15}{2}(\dxxf)^2 \dxf
\right],\\
\partial_{xxk} L(x,x) &= \frac{\E^{f(x,x)}}{\varsigma\sqrt{2\pi}} \left(\frac{1}{32} + \frac{1}{8\varsigma^2} \right),
\end{align*}
and one last time to reach
\begin{align*}
\partial_{xxxk} L & =  \frac{\E^f}{\varsigma\sqrt{2\pi}} \bigg[\frac34 (\dxf)^4 - \frac54 (\dxf)^5 + \half (\dxf)^6 + \frac92 \dxxf (\dxf)^2 - \frac{25}{2} \dxxf (\dxf)^3 \\
&\qquad - \frac{75}{4} (\dxxf)^2 \dxf + \frac{15}{2} \dxxf (\dxf)^4 + \frac{45}{2} (\dxxf)^2(\dxf)^2 + \frac94 (\dxxf)^2 + \frac{15}{2} (\dxxf)^3
\bigg],\\
\partial_{xxxk} L(x,x) & =  \frac{\E^{f(x,x)}}{\varsigma\sqrt{2\pi}} \bigg(\frac{1}{64} - \frac{1}{32\varsigma^2} - \frac{3}{2\varsigma^4} - \frac{15}{2\varsigma^6} \bigg).
\end{align*}
We can conclude that:
\begin{align*}
H(x,x) = (\partial_{xxxk} - \partial_{xxk})L(x,x)= \frac{\E^{f(x,x)}}{\varsigma\sqrt{2\pi}} \left(-\frac{1}{64} - \frac{5}{32\varsigma^2} - \frac{3}{2\varsigma^4}- \frac{15}{2\varsigma^6} \right).
\end{align*}


\bibliographystyle{siam}
\bibliography{Bibliography}

\end{document}